\documentclass[a4paper,onecolumn,11pt,accepted=2026-06-20]{quantumarticle}
\pdfoutput=1
\usepackage{caption}
\usepackage{amsmath,amssymb,amsthm,color,bm,mathrsfs,extarrows,mathtools,enumitem,fancybox,makecell}
\usepackage[square,sort,comma,numbers]{natbib}
\usepackage{xcolor,bbm}
\usepackage{comment}
\usepackage{soul}
\usepackage{ifthen}
\usepackage{mathdots}
\usepackage{enumitem}
\usepackage{multirow}
\usepackage{tablefootnote}
\usepackage{thm-restate}

\allowdisplaybreaks

\setlist[itemize]{itemsep=0pt}
\setlist[enumerate]{itemsep=0pt}

\usepackage{hyperref}
\usepackage[capitalise,nameinlink]{cleveref}
\Crefname{lemma}{Lemma}{Lemmas}
\Crefname{fact}{Fact}{Facts}
\Crefname{theorem}{Theorem}{Theorems}
\Crefname{corollary}{Corollary}{Corollaries}
\Crefname{claim}{Claim}{Claims}
\Crefname{example}{Example}{Examples}
\Crefname{problem}{Problem}{Problems}
\Crefname{definition}{Definition}{Definitions}
\Crefname{notation}{Notation}{Notations}
\Crefname{assumption}{Assumption}{Assumptions}
\Crefname{subsection}{Subsection}{Subsections}
\Crefname{section}{Section}{Sections}
\Crefformat{equation}{(#2#1#3)}

\newtheorem{theorem}{Theorem}[section]
\newtheorem*{theorem*}{Theorem}

\newtheorem{proposition}[theorem]{Proposition}
\newtheorem*{proposition*}{Proposition}

\newtheorem*{property*}{Property}
\newtheorem{lemma}[theorem]{Lemma}
\newtheorem*{lemma*}{Lemma}
\newtheorem{corollary}[theorem]{Corollary}
\newtheorem*{corollary*}{Corollary}
\newtheorem*{conjecture*}{Conjecture}
\newtheorem{fact}[theorem]{Fact}
\newtheorem*{fact*}{Fact}

\newtheorem*{exercise*}{Exercise}

\newtheorem*{hypothesis*}{Hypothesis}

\theoremstyle{definition}
\newtheorem{definition}[theorem]{Definition}

\newtheorem{exercise-easy}[theorem]{Exercise}
\newtheorem{exercise-med}[theorem]{Exercise}
\newtheorem{exercise-hard}[theorem]{Exercise$^\star$}
\newtheorem{claim}[theorem]{Claim}
\newtheorem*{claim*}{Claim}

\newtheorem{remark}[theorem]{Remark}
\newtheorem*{remark*}{Remark}

\newtheorem*{observation*}{Observation}

\DeclareSymbolFont{extraup}{U}{zavm}{m}{n}
\DeclareMathSymbol{\varheart}{\mathalpha}{extraup}{86}
\DeclareMathSymbol{\vardiamond}{\mathalpha}{extraup}{87}

\DeclareMathOperator*{\E}{\mathbb E}

\newcommand{\eps}{\varepsilon}
\newcommand{\abs}[1]{\left| #1 \right|}
\newcommand{\vabs}[1]{\left\| #1 \right\|}
\newcommand{\abra}[1]{\left\langle #1 \right\rangle}
\newcommand{\pbra}[1]{\left( #1 \right)}
\newcommand{\sbra}[1]{\left[ #1 \right]}
\newcommand{\cbra}[1]{\left\{ #1 \right\}}
\newcommand{\floorbra}[1]{\left\lfloor #1 \right\rfloor}
\newcommand{\ceilbra}[1]{\left\lceil #1 \right\rceil}

\newcommand{\hsnorm}[1]{\vabs{ #1 }_{\mathrm{HS}}}

\newcommand{\ket}[1]{| #1 \rangle}
\newcommand{\bra}[1]{\langle #1 |}
\newcommand{\braket}[2]{\left\langle #1 \middle| #2\right\rangle}
\newcommand{\ketbra}[2]{\left| #1 \right\rangle\left\langle #2\right|}
\newcommand{\bin}{\{0,1\}}
\newcommand{\binpm}{\{\pm1\}}

\newcommand{\poly}{\mathsf{poly}}

\newcommand{\tr}{\mathrm{Tr}}

\newcommand{\diag}{\mathsf{diag}}

\newcommand{\Ham}{\mathsf{Ham}}

\newcommand{\Acal}{\mathcal{A}}
\newcommand{\Bcal}{\mathcal{B}}
\newcommand{\Ccal}{\mathcal{C}}

\newcommand{\Ecal}{\mathcal{E}}

\newcommand{\Gcal}{\mathcal{G}}
\newcommand{\Hcal}{\mathcal{H}}

\newcommand{\Scal}{\mathcal{S}}
\newcommand{\Tcal}{\mathcal{T}}
\newcommand{\Ucal}{\mathcal{U}}

\newcommand{\Tgate}{{T}}
\newcommand{\Hgate}{{H}}
\newcommand{\Xgate}{{X}}
\newcommand{\Ygate}{{Y}}
\newcommand{\Zgate}{{Z}}
\newcommand{\Stab}{\mathsf{Stab}}

\renewcommand{\tilde}{\widetilde}
\renewcommand{\bar}{\overline}

\title{Quantum State Preparation with Optimal T-Count}

\author{David Gosset}
\affiliation{Google Quantum AI}
\affiliation{Department of Combinatorics and Optimization and Institute for Quantum Computing, University of Waterloo}
\affiliation{Perimeter Institute for Theoretical Physics}
\orcid{0000-0003-3975-2253}

\author{Robin Kothari}
\affiliation{Google Quantum AI}
\orcid{0000-0001-6114-943X}

\author{Kewen Wu}
\affiliation{Computing and Mathematical Sciences Department, Caltech}
\orcid{0000-0002-5894-822X}

\date{}

\begin{document}

\maketitle

\begin{abstract}
How many $T$ gates are needed to approximate an arbitrary $n$-qubit quantum state to within error $\varepsilon$? Improving prior work of Low, Kliuchnikov, and Schaeffer, we show that the optimal asymptotic scaling is 
$$
\Theta\left(\sqrt{2^n\log(1/\varepsilon)}+\log(1/\varepsilon)\right)
$$
if we allow ancilla qubits. We also show that this is the optimal $T$-count for implementing an arbitrary diagonal $n$-qubit unitary to within error $\varepsilon$. We describe applications in which a tensor product of many single-qubit unitaries can be synthesized in parallel for the price of one.
\end{abstract}

\vspace{1ex}
\tableofcontents 

\section{Introduction}\label{sec:intro}

The Gottesman-Knill theorem \cite{gottesman1998heisenberg} reveals a remarkable corner of Hilbert space occupied by Clifford circuits and stabilizer states --- classically efficiently simulable but nevertheless equipped with genuinely quantum features such as many-body entanglement and superposition. 
Only non-Clifford operations can move us out of this corner, a requirement for universal quantum computation.  
A standard choice is to compile quantum circuits using the Clifford+$\Tgate$ gate set, in which only the single-qubit $\Tgate=\diag(1,e^{i\pi/4})$ gate\footnote{We use $\diag(c_1,\ldots,c_m)$ to denote the diagonal matrix with $c_1,\ldots,c_m$ on the diagonal in order.} is non-Clifford.

In recent years there has been an interest in quantifying how much non-Cliffordness, aka \textit{magic}, is required to implement quantum states and operations (see e.g., \cite{veitch2014resource, bravyi2016trading, bravyi2019simulation, leone2022stabilizer}). This directly relates to the cost of fault-tolerant quantum computation based on 2D stabilizer error correcting codes such as the surface code, where the implementation of $\Tgate$ gates using magic state distillation is vastly more costly than the transversal implementation of Clifford gates~\cite{bravyi2005universal}.\footnote{A recent line of research \cite{litinski2019magic,gidney2024magic} shows that $\Tgate$ gates might not be as expensive as previously thought for certain architectures and fault tolerance schemes.} A second reason to study quantum magic is its relevance to classical simulation of quantum computers: extensions of the Gottesman-Knill theorem give rise to classical simulation algorithms with a cost that scales polynomially in all parameters except the number of non-Clifford gates \cite{bravyi2016trading, bravyi2016improved, bravyi2019simulation}. Another line of work has explored the role that magic plays in quantum many-body physics and phases of matter \cite{liu2022many,ellison2021symmetry}.

So how much magic is needed to make a quantum state, in the worst case? In particular, let us consider the number of $\Tgate$ gates, or \textit{$\Tgate$-count}, required to $\eps$-approximate an arbitrary $n$-qubit state to within $\eps$-error in the $\ell_2$ norm. 
Here we consider unitary state preparation circuits composed of Clifford and T gates.

For $n=1$,  a remarkable series of works incorporated techniques from algebraic number theory to give a fairly complete answer to this question \cite{kliuchnikov2013fast,selinger2015efficient,ross2016optimal}. The algorithm of Ross and Selinger \cite{ross2016optimal} computes a sequence of\footnote{In this paper, asymptotic notations $O(\cdot),\Omega(\cdot)$ only hide absolute constants that do not depend on any parameter, and $\log(x)$ is the base 2 logarithm of $x$.} $O(\log(1/\eps))$ single-qubit $\Hgate$ and $\Tgate$ gates that  $\eps$-approximates a given single-qubit unitary $U$ (or prepares a $1$-qubit state $U|0\rangle$), matching information-theoretic lower bounds \cite{harrow2002efficient, beverland2020lower}. These methods leverage beautiful structural properties of the group generated by single-qubit Clifford and $\Tgate$ gates established in \cite{kliuchnikov2013fast} to outperform the well-known approximate synthesis technique based on the Solovay-Kitaev theorem \cite{nielsen2010quantum}.

These methods can also be used to prepare an $n$-qubit state; one can first express it as $U|0^n\rangle$ where $U$ is a sequence of $O(2^n)$ CNOT and single-qubit gates (see e.g., \cite{barenco}), and then use the Ross-Selinger algorithm~\cite{ross2016optimal} to approximate each of the single-qubit gates as a sequence of $\Hgate$ and $\Tgate$ gates. In this way one can upper bound the number of $\Tgate$ gates for $n$-qubit state preparation as $O(2^n(n+\log(1/\eps)))$. Note that this is also an upper bound on the total number of gates used. Surprisingly, \cite[Section 3]{low2018trading} showed that the $\Tgate$-count of $n$-qubit state synthesis can be improved by almost a square-root factor to 
\begin{equation}
O\left(\sqrt{2^n n\log(n/\eps)}+\log^2(n/\eps)\right),
\label{eq:lkscost}
\end{equation}
if we allow the use of ancilla qubits that are prepared and returned to the all-zeros state.\footnote{It is possible to tighten their analysis to give better (but still sub-optimal) results. See \Cref{sec:improve_LKS} for details.} 

The synthesis algorithm of \cite{low2018trading} has two important ingredients. The first ingredient is a state preparation method proposed by Grover and Rudolph \cite{grover2002creating}, which we now describe. For simplicity suppose the target state $\ket\psi=\sum_{x\in\bin^n}\alpha_x\ket x$ has non-negative real amplitudes $\alpha_x\ge0$.
Then, starting with $\ket{0^n}$, we first apply a single-qubit rotation on the first qubit to ensure it has the correct marginal, i.e., $\ket0\to\sqrt{\sum_{x:x_1=0}\alpha_x^2}\ket0+\sqrt{\sum_{x:x_1=1}\alpha_x^2}\ket1$.
Then, controlled on the first qubit, we rotate the second qubit to have the correct marginal.  We repeat this process until the $n$-th qubit. This reduces state preparation to the task of synthesizing $n$ multiply-controlled single-qubit unitaries, or (using Euler angles) $3n$ \textit{diagonal} multiply-controlled single-qubit unitaries. The second ingredient is a method for approximating diagonal controlled single-qubit unitaries that starts with a $b=O(\log(n/\eps))$-bit approximation of the rotation angle and then implements a controlled rotation corresponding to each bit. The technique used to implement the latter controlled rotations can be viewed as a generalization of the fact that an $n$-qubit \textit{Boolean} diagonal unitary, with nonzero entries in $\pm 1$,  can be implemented with $\Tgate$-count $\Theta\pbra{\sqrt{2^n}}$ \cite{nechiporuk1962complexity,schnorr1989multiplicative,maslov2016optimal,low2018trading}.

If we look beyond the headline square-root dependence $\sqrt{2^n}$, we can identify contributions to the $\Tgate$-count from each of these two ingredients that we might hope to improve: an overhead from the Grover-Rudolph method that scales with the number of qubits $n$ (under the square root in the first term of \Cref{eq:lkscost}) as well as an overhead that scales as $\log^2(n/\eps)$ that arises from implementing $b$ controlled single-qubit rotations, each to within $\eps$-error.

In this work we show that indeed both of these overheads can be avoided. Our main result is the following improvement which gives an optimal upper bound on the $\Tgate$-count of an $n$-qubit state.

\begin{restatable}[Quantum state preparation with optimal $\Tgate$-count]{theorem}{stateT}\label{thm:state_T-count}
Any $n$-qubit state can be prepared up to error $\eps$ by a Clifford+$\Tgate$ circuit starting with the all-zeros state using 
\begin{equation}
\label{eq:optimalT}
O\pbra{\sqrt{2^n\log(1/\eps)}+\log(1/\eps)}    
\end{equation}
$\Tgate$ gates and ancillas.\footnote{This construction uses $O(2^n\log(1/\eps))$ Clifford gates.}
Furthermore, no Clifford+$\Tgate$ circuit (even with measurements and adaptivity) can use asymptotically fewer $\Tgate$ gates.
\end{restatable}

More precisely, for any $n$-qubit state $\ket\psi$, we show there is a Clifford+$\Tgate$ circuit $U$ such that $U\ket{0^n}\ket{0^a}=\ket{\tilde\psi}\ket{0^a}$, where $\| \ket\psi - \ket{\tilde\psi}\| \leq \eps$. The $\Tgate$-count of $U$ and $a$, the number of ancillas used, are upper bounded by \Cref{eq:optimalT}.

The lower bound in \Cref{thm:state_T-count}, which we formally establish in \Cref{thm:state_adaptive_lb}, holds in an even stronger model that we call \emph{adaptive} Clifford+$\Tgate$ circuits (defined formally in \Cref{sec:lower_bounds}). In this model, the algorithm may apply mid-circuit measurements and any future gates may depend on the outcomes of prior measurements. Such a circuit naturally outputs a mixed state, and even if we only require this state to be $\eps$-close to the target state in trace distance, we show that the lower bound on $\Tgate$-count in \Cref{eq:optimalT} still holds. 
The lower bound is established combining the $\Omega(\log(1/\eps))$ lower bound from \cite{beverland2020lower}, which holds for adaptive Clifford+$\Tgate$ circuits, and a slight generalization of the $\Omega\pbra{\sqrt{2^n\log(1/\eps)}}$ lower bound in \cite{low2018trading}, which only holds for unitary Clifford+$\Tgate$ circuits.\footnote{We thank Luke Schaeffer for confirming that their $\Omega\pbra{\sqrt{n2^n\log(1/\eps)}}$ lower bound at the end of \cite[Section 5]{low2018trading} should actually be $\Omega\pbra{\sqrt{2^n\log(1/\eps)}}$. }
Thus  \Cref{thm:state_T-count} represents the best of both worlds, with the upper bound being established in the weak model, but the lower bound being established in the strong model.

The circuit in \Cref{thm:state_T-count} is obtained by refining each of the two ingredients in the LKS construction to avoid the overheads described above. We now describe our refinements. 

For the first ingredient, we use a variant of recent methods from \cite{irani2022quantum, rosenthal2024efficient} which reduce quantum state synthesis to the synthesis of $O(1)$ diagonal unitaries --- far fewer than are needed in the Grover-Rudolph method. We start with a coarse approximation: we use just two \textit{Boolean} diagonal unitaries to construct a state $\ket\phi$ that has a constant overlap with the target state $\ket\psi$ (i.e., $\eps$ is a constant). Specifically, we show that there exist two diagonal unitaries $B_1,B_2$ with diagonal entries in $\pm 1$, such that $\ket\phi:=B_2\Hgate^{\otimes n}B_1\Hgate^{\otimes n}\ket{0^n}$ has a constant overlap with $\ket\psi$. Starting from this coarse approximation, we use an idea from \cite{rosenthal2024efficient} to improve the accuracy to the desired value of $\eps$. Since $\ket\phi$ already has a constant overlap with $\ket\psi$, intuitively their difference $\ket\psi-\ket\phi$ should have a small length. Then our goal becomes to construct $\ket\psi-\ket\phi$, for which we again have a coarse approximation. Iterating this, we obtain $\ket\psi$ in the limit, where the error decays exponentially: $\ket\psi$ is $\eps$-close to
\begin{equation}
\ket{\psi'}:=\beta\cdot\sum_{k=0}^{O(\log(1/\eps))}\gamma^k\ket{\phi_k}
\label{eq:overview_3}
\end{equation}
for some constants $\beta$ and $\gamma$, where each $\ket{\phi_k}=B_2^{(k)}\Hgate^{\otimes n}B_1^{(k)}\Hgate^{\otimes n}\ket{0^n}$ is a coarse approximation of the previous step. To prepare this state, we use the LCU (linear combination of unitaries) method along with amplitude amplification, see \Cref{sec:flag_and_amplify} for details.

For the second ingredient, we give a new method to synthesize diagonal unitaries with optimal $\Tgate$-count. 

\begin{restatable}[Diagonal unitary synthesis with optimal $\Tgate$-count]{theorem}{diagonalunitary}\label{thm:diagonal_unitary_T-count}
Any diagonal unitary on $n$ qubits can be implemented up to error $\eps$ by a Clifford+$\Tgate$ circuit using 
\begin{equation}
O\pbra{\sqrt{2^n\log(1/\eps)}+\log(1/\eps)}     
\end{equation}
$\Tgate$ gates and ancillas.\footnote{This construction uses $O(2^n\log(1/\eps))$ Clifford gates.}
Furthermore, no Clifford+$\Tgate$ circuit (even with measurements and adaptivity) can use asymptotically fewer $\Tgate$ gates.
\end{restatable}

Here, the distance between two unitaries is the operator norm (aka spectral norm) of their difference, i.e., $U$ implements $V$ to error $\eps$ if $\|U-V\|=\eps$. 
The lower bound in \Cref{thm:diagonal_unitary_T-count}, which is formally proved in \Cref{thm:diagonal_adaptive_lb}, holds in the stronger model of adaptive Clifford+$\Tgate$ circuits, as before. The lower bound is proved by a reduction to the previous lower bound; see \Cref{sec:lower_bounds} for details.

To prove the upper bound from \Cref{thm:diagonal_unitary_T-count}, we first note that an $n$-qubit diagonal unitary $D$ can be regarded as a single-qubit diagonal unitary $G_y$ controlled on $n-1$ qubits $\ket y$. Now we use the fact that any single-qubit unitary $G_y$ can be approximated by a sequence of $O(\log(1/\eps))$ Hadamard and $\Tgate$ gates (using no ancillas) \cite{kliuchnikov2013fast,selinger2015efficient,kliuchnikov2015practical,ross2016optimal}. To synthesize $D$, we start by applying a Boolean oracle $B\colon\ket y\ket{z} \ket 0\to\ket y\ket{z}\ket{s_y}$, where $y\in \{0,1\}^{n-1}, z\in \{0,1\}$, and $s_y$ is a binary string of length $O(\log(1/\eps))$ that describes a sequence of $\Hgate$ and $\Tgate$ gates that $\eps$-approximates $G_y$. This Boolean oracle $B$ can be implemented with $\Tgate$-count $O\pbra{\sqrt{2^n\log(1/\eps)}}$ \cite{maslov2016optimal,low2018trading}. After this step, we apply a sequence of controlled-Hadamard and controlled-$\Tgate$ gates on the $n$-th qubit, each controlled on a bit of $\ket{s_y}$. The total $\Tgate$-cost of this last step is $O(\log(1/\eps))$, proportional to the length of $s_y$. Finally, we uncompute the Boolean oracle $B$, returning the third register to the all-zeros state.

We note that it is possible to directly improve on the state synthesis technique from \cite{low2018trading} using \Cref{thm:diagonal_unitary_T-count}; this gives a better $\Tgate$-count than \Cref{eq:lkscost}, but still worse than the optimal bound \Cref{eq:optimalT}. See \Cref{sec:improve_LKS} for details.

\subsection{Applications}

There are some immediate gains by using our bounds in place of those from \cite{low2018trading} in known applications. 
For instance, we obtain slight improvements to the state-of-the-art lower bound on the stabilizer rank of magic states, due to Mehraban and Tahmasbi \cite{mehraban2024quadratic}.\footnote{In particular, we can shave off two factors of $\log m$ to get a slightly improved bound for the approximate stabilizer rank $\chi_\delta(|\Tgate\rangle^{\otimes m})=\Omega\pbra{\frac{m^2}{\log^2 m}}$ for $\delta=\Omega(1)$.}  
Our results also improve some related lower bounds concerning decompositions of Boolean functions as linear combination of quadratic functions \cite{filmus2014real}; see \cite{mehraban2024quadratic} for details. We can also obtain improved bounds for ``partial unitary synthesis'' where the goal is to implement the first $K$ columns of an $n$-qubit unitary. Here state preparation corresponds to the case of $K=1$ and unitary synthesis is the case of $K=2^n$.
Using Householder reflections \cite{householder1958unitary}, this can be reduced to the task of preparing up to $K$ different $n$-qubit states \cite{kliuchnikov2013synthesis, low2018trading}. Via this reduction and with our \Cref{thm:state_T-count}, we obtain a $\Tgate$-count of $O\pbra{K\sqrt{2^n\log(K/\eps)}+K\log(K/\eps)}$, which similarly shaves polylogarithmic factors from the bounds in \cite{low2018trading}.

Below we discuss two further applications of \Cref{thm:diagonal_unitary_T-count} to synthesizing tensor products of single-qubit unitaries.

\paragraph*{Batched Synthesis of Single-Qubit Unitaries.}
Results from \cite{selinger2015efficient,beverland2020lower} show that synthesizing a single-qubit unitary has $\Tgate$-count $\Theta(\log(1/\eps))$.
What is the cost of implementing $m$ single-qubit unitaries $U_1 \otimes \cdots \otimes U_m$ to error $\eps$?\footnote{Here $\eps$ is the operator-norm error in approximating the entire unitary $U_1 \otimes \cdots \otimes U_m$.}
The most natural bound would be $O(m\log(m/\eps))$ by implementing each up to $\eps/m$ error.
Surprisingly, a corollary of our previous result implies that the total cost remains the same $\log(1/\eps)$ even if $m$ has slightly superconstant dependence on $1/\eps$.

\begin{restatable}{theorem}{batchedsyn}\label{thm:batched_synthesis}
Let $m= O(\log\log(1/\eps))$ be an integer.
A tensor product of $m$ (potentially different) single-qubit unitaries $U_1 \otimes \cdots \otimes U_m$ can be implemented up to error $\eps$ by a Clifford+$\Tgate$ circuit using $O\pbra{\log(1/\eps)}$ $\Tgate$ gates and ancillary qubits.\footnote{This construction uses $O(\log(1/\eps))$ Clifford gates.}
\end{restatable}

\Cref{thm:batched_synthesis} is a simple consequence of \Cref{thm:diagonal_unitary_T-count}. To see this, note that we can use the Euler angle decomposition to express each single-qubit unitary $U_j$ as a sequence of Hadamards and at most $3$ diagonal unitaries, i.e, $U_j=D_{3,j}\Hgate D_{2,j}\Hgate D_{1,j}$, where $D_{i,j}$ are diagonal single-qubit unitaries for $i\in \{1,2,3\}$ and $j\in [m]$. This reduces our synthesis task to the case in which all the single-qubit unitaries are diagonal; clearly it is enough to synthesize $D_{i,1}\otimes D_{i,2}\otimes \cdots\otimes D_{i,m}$ for $i=1,2,3$.  But now we can apply \Cref{thm:diagonal_unitary_T-count} which implies that any $\log\log(1/\eps)$-qubit diagonal unitary can be implemented with $\Tgate$-count $O(\log(1/\eps))$.

An interesting direction for future work would be to improve on this batched synthesis result (or show that improvements are not possible).  In particular, it is an open question whether or not an exponentially strengthened version of the above theorem holds with $m=\Omega(\log(1/\eps))$.

\paragraph*{Mass Production of a Single-Qubit Unitary.}
The batched synthesis result (\Cref{thm:batched_synthesis}) focused on implementing $m$ \emph{different} single-qubit unitaries.
It is natural to ask what happens if we want to prepare $m$ \emph{identical} single-qubit unitaries.
This is called ``mass production'' \cite{ulig1974synthesis,uhlig1992networks,kretschmer2023quantum} in circuit synthesis and lives in a richer family of direct sum problems; see the survey \cite{pankratov2012direct} for details.
As an application of \Cref{thm:diagonal_unitary_T-count}, we prove the following optimal $\Tgate$-count for mass production of a single-qubit unitary.

\begin{restatable}{theorem}{massprod}\label{thm:mass_production}
Let $U$ be an arbitrary single-qubit unitary.
Then $U^{\otimes m}$ can be implemented up to error $\eps$ by a Clifford+$\Tgate$ circuit using $O\pbra{m+\log(1/\eps)}$ $\Tgate$ gates and ancillary qubits.\footnote{This construction uses $O(m\log(1/\eps))$ Clifford gates.}
Furthermore, no Clifford+$\Tgate$ circuit (even with measurements and adaptivity) can use asymptotically fewer $\Tgate$ gates.
\end{restatable}

The bound in \Cref{thm:mass_production} is indeed better than the one in \Cref{thm:batched_synthesis} if the goal is to implement the same unitary several times: \Cref{thm:mass_production} shows that using $O(\log(1/\eps))$ $\Tgate$ gates, we can prepare $m=O(\log(1/\eps))$ copies of a given single-qubit unitary $U$, 
in contrast to \Cref{thm:batched_synthesis}, which says we can implement $m= O(\log\log(1/\eps))$ (potentially different) unitaries.

 \Cref{thm:mass_production} is also a corollary of \Cref{thm:diagonal_unitary_T-count}. Using the Euler angle decomposition, we can reduce to the case where $U$ is diagonal, say, $U=\diag(1,e^{i\theta})$. Then $U^{\otimes m}\ket x=e^{i\theta\cdot h_x}\ket x$, where $h_x\in\cbra{0,1,\ldots,m}$ is the Hamming weight of $x\in\bin^m$.
Therefore, we can first compute the Hamming weight into a second register containing $L=\lceil \log m\rceil$ qubits $\ket x\to\ket x\ket{h_x}$, then apply the corresponding phase change $\ket{h_x}\to e^{i\theta\cdot h_x}\ket{h_x}$. The first step can be done with $O(m)$ complexity in a standard way (see \Cref{fct:hamming_weight_circuit} for details). For the second step, since $h_x$ has $L$ bits, the phase change is a diagonal unitary on $L=\ceilbra{\log m}$ qubits, which can be implemented with $\Tgate$-count $O\pbra{\sqrt{m\log(1/\eps)}+\log(1/\eps)}\le O\pbra{m+\log(1/\eps)}$ by \Cref{thm:diagonal_unitary_T-count}.

Note that one can implement a similar but sub-optimal strategy without the use of \Cref{thm:diagonal_unitary_T-count}, see, e.g., \cite[Section 3.2]{beverland2020lower}. Indeed, one can implement the second step as $R(\theta)\otimes R(2\theta)\otimes\ldots R(2^{L}\theta) \ket{h_x}= e^{i\theta\cdot h_x}\ket{h_x}$ where $R(\alpha)\equiv\diag(1,e^{i\alpha})$ and then use the Ross-Selinger algorithm to approximate each single-qubit gate $R(2^k\theta)$ to within error $\eps/L$. This strategy gives an overall $\Tgate$-count of 
\begin{align}
O(m+L\cdot \log(L/\eps))
&=O(m+\log(m)\log(\log(m)/\eps))\\
&=O(m+\log(m)\log(1/\eps)),    
\end{align}
which is worse than the optimal bound given above.

Finally we remark that it is easy to see that the $\Tgate$-count upper bound $O(m+\log(1/\eps))$ from \Cref{thm:mass_production} is asymptotically tight. On the one hand, the $\Omega(\log(1/\eps))$ $\Tgate$-count lower bound holds even if $m=1$ \cite{beverland2020lower}.
On the other hand, if $U$ is the $\Tgate$ gate, we cannot hope to compress $\Tgate^{\otimes m}$ in a generic way, which otherwise contradicts the optimality of \Cref{thm:state_T-count}.
See details in \Cref{sec:mass}.

\paragraph*{Discussion.}
In this work we have characterized the worst-case $\Tgate$-count of quantum states and diagonal unitaries. A challenging open question is to understand the optimal $\Tgate$-count for general $n$-qubit unitaries. 
We can prove an $\tilde\Omega(2^n)$ $\Tgate$-count lower bound (see \Cref{thm:unitary_adaptive_lb} for details).
However there is a glaringly large gap between this lower bound and the best known upper bound $\tilde O\pbra{2^{1.5n}}$ \cite{low2018trading, rosenthal2021query}. This upper bound is obtained in \cite{low2018trading} by reducing the unitary synthesis task to that of preparing $2^n$ $n$-qubit states, each of which uses $\tilde O\pbra{\sqrt{2^n}}$ $\Tgate$ gates; \cite{rosenthal2021query} also shows how to synthesize a unitary efficiently with $O\pbra{\sqrt{2^n}}$ calls to a $2n$-qubit Boolean oracles, each of which needs $O(2^n)$ $\Tgate$ gates (see \Cref{lem:boolean_oracle}).

While our focus is on $T$-count and our construction gives the asymptotic optimal bound, the number of Clifford gates in our state preparation circuit is $O(2^n\log(1/\eps))$, which can be verified given the explicit description of our synthesis procedure. This bound is close to the optimal $\Omega(2^n/n)$ lower bound. In particular, if the goal is to minimize total gate count, we can replace \Cref{lem:boolean_oracle} and \Cref{rmk:boolean} with the construction that achieves the $O(2^n/n)$ gate count for Boolean functions; then our final circuit will have $O(2^n/n)$ gates, albeit a similar amount of $T$ gates.
It remains an open problem whether it is possible to achieve $O(\sqrt{2^n})$ $T$-count and $O(2^n/n)$ total gate count at the same time.

We emphasize that our savings on the $T$-count is made possible with the use of ancillary qubits, which is also necessary by a simple counting argument.
By tuning the construction for \Cref{lem:boolean_oracle} with different number of ancillas and $T$-count, we believe our construction achieves near optimal tradeoffs between ancillas and $T$-count (with some logarithmic slack in $n$ and $1/\eps$). The best achievable tradeoff between these two metrics is still unknown.

A similarly interesting question is to study the $T$-depth in place of $T$-count. The most depth-consuming part in our construction is the synthesis of (Boolean) diagonal unitaries, which, in light of \Cref{rmk:boolean}, should be $\poly(n)$. The optimal $T$-depth is still open and we expect that the lower bound side will be significantly more challenging, if given unbounded ancillas.

\paragraph*{Subsequent Work.}
A recent followup by Tan \cite{tan2025unitary} showed how to improve the unitary synthesis $T$-count to $\tilde O(2^{4n/3})$. This is a significant improvement over the aforementioned $\tilde O(2^{1.5n})$ bound, but is still far from the $\tilde\Omega(2^n)$ lower bound.

\paragraph*{Paper Organization.}
The remainder of the paper is organized as follows. In \Cref{sec:diagonal_T}, we establish the optimal $\Tgate$-count for diagonal unitaries (\Cref{thm:diagonal_unitary_T-count}) and in \Cref{sec:batch} and \Cref{sec:mass} we describe the applications to batched synthesis (\Cref{thm:batched_synthesis}) and mass production (\Cref{thm:mass_production}). In \Cref{sec:state_T}, we bring in other techniques and prove the optimal $\Tgate$-count for state preparation (\Cref{thm:state_T-count}). Detailed proofs of lower bounds for \Cref{thm:state_T-count} (as \Cref{thm:state_adaptive_lb}) and \Cref{thm:diagonal_unitary_T-count} (as \Cref{thm:diagonal_adaptive_lb}) are presented in \Cref{sec:lower_bounds}.
A canonical form for Clifford+$\Tgate$ circuit with Pauli postselections is proved in \Cref{sec:canonical_form}, which generalizes \cite{beverland2020lower}.
An improved analysis of \cite{low2018trading} is in \Cref{sec:improve_LKS}. 


\section{Diagonal unitary synthesis with optimal T-count}\label{sec:diagonal_T}

In this section, we prove \Cref{thm:diagonal_unitary_T-count} (restated below) and give details of the applications described in the previous section.

\diagonalunitary*

In \Cref{thm:diagonal_unitary_T-count}, the diagonal entries of the diagonal unitary $D$ can be arbitrary complex numbers of magnitude 1. 
An important special case is obtained by restricting to ``Boolean phase oracles'': diagonal unitaries with $\pm1$ diagonal entries. Such unitaries can be implemented using Clifford+$\Tgate$ circuits with zero error. Moreover, it is known how to achieve this using very few $\Tgate$ gates  \cite{nechiporuk1962complexity,schnorr1989multiplicative,maslov2016optimal,low2018trading}.
We shall use the following statement which is phrased in terms of the standard oracle that computes the Boolean function into an ancilla register.

\begin{lemma}[{\cite[Theorem 2]{low2018trading}}]\label{lem:boolean_oracle}
Let $b\ge1$ be an integer and $f\colon\bin^n\to\bin^b$ be an arbitrary Boolean function.
Define $U_f$ as the unitary mapping $\ket x\ket y$ to $\ket x\ket{y\oplus f(x)}$ for all $x\in\bin^n,y\in\bin^b$.
Then $U_f$ can be implemented exactly by a Clifford+$\Tgate$ circuit using $O\pbra{\sqrt{b\cdot2^n}}$ $\Tgate$ gates and ancillary qubits.\footnote{This construction uses $O(b\cdot2^n)$ Clifford gates.}
\end{lemma}
\begin{remark}\label{rmk:boolean}
We sketch the idea to prove \Cref{lem:boolean_oracle}.
Assume $f$ has a classical circuit consisting of XOR gates and AND gates with bounded fan-in.
By converting XOR gate into CNOT gate (which is Clifford) and AND gate into Toffoli gate (which has a constant $\Tgate$-count), we can synthesize $U_f$ with $\Tgate$-count proportional to the number of AND gates.
The latter quantity is termed ``multiplicative complexity'' in the classical complexity theoretic literature and is thoroughly studied in e.g., \cite{nechiporuk1962complexity,schnorr1989multiplicative}; here we also give a brief overview of the upper bound. 

Assume $b=1$ for simplicity.
One way to compute $f$ is to gradually fix each input bit: $f(x)=\pbra{x_1\land f_1(x_2,\ldots,x_n)}\oplus f_2(x_2,\ldots,x_n)$, where $f_2$ (resp., $f_1\oplus f_2$) is $f$ with $x_1$ fixed to $0$ (resp., $1$).
Let $d\in[n]$ be an integer to be optimized later.
If we iterate the above expansion for $x_1,\ldots,x_d$, then we need $2^d$ AND gates and it remains to compute $2^d$ Boolean functions on $x_{d+1},\ldots,x_n$.
Now observe that once all conjunctions $\bigwedge_{j\in S}x_j$, for $S\subseteq\cbra{d+1,\ldots,n}$ are computed (which takes $2^{n-d}$ AND gates\footnote{To see this, note that each $|S|=k$ is an AND of some size-$(k-1)$ conjunction and one extra variable. Therefore, building on size-$(<k)$ conjunctions, we only need to introduce $\binom{n-d}k$ AND gates to get size-$(\le k)$ conjunctions. So the total cost is $\sum_{k\ge2}\binom{n-d}k<2^{n-d}$.}), we just need to XOR some of them for each one of the remaining $2^d$ Boolean functions.
Therefore, the total number of AND gates is $2^d+2^{n-d}$, which has the claimed minimum value  $O(\sqrt{2^n})$ by setting $d=n/2$.
In the case of $b>1$, an analogous calculation shows the total number of AND gates is roughly $b\cdot2^d+2^{n-d}$, which has minimum value $O(\sqrt{b\cdot2^n})$ as claimed.

Note that the number of XOR gate used is $O(b\cdot2^n)$.
\end{remark}

We shall also use the known $\Tgate$-optimal synthesis of single-qubit unitaries without ancillas \cite{kliuchnikov2013fast,selinger2015efficient,kliuchnikov2015practical,ross2016optimal}. 

\begin{lemma}[{\cite{selinger2015efficient,ross2016optimal}}]\label{lem:single_qubit_T-count}
Any single-qubit unitary with determinant $1$\footnote{As noted in \cite{selinger2015efficient}, this determinant condition is needed since a Clifford+$\Tgate$ circuit has determinant in $\{e^{i\cdot t\pi/4}\colon t=0,1,\ldots,7\}$. This is not a problem in practice because the determinant can be adjusted by multiplying the unitary with a global phase.} can be implemented up to error $\eps$ by a Clifford+$\Tgate$ circuit using $O(\log(1/\eps))$ $\Tgate$ gates and without ancillary qubits.\footnote{In particular, the number of Clifford gates is also $O(\log(1/\eps))$.}
\end{lemma}

Below we give the proof of \Cref{thm:diagonal_unitary_T-count}. The idea is to view the diagonal unitary $D$ as a single-qubit unitary acting on the $n$th qubit, controlled on the state of the first $n-1$ qubits.  For each fixed value of the control bits, we have a single-qubit unitary that can be approximated by a sequence\footnote{Since \Cref{lem:single_qubit_T-count} uses no ancilla, the Clifford gates are only Hadamard $\Hgate$ and Phase gate $S=T^2$. Hence we simply assume the sequence contains only Hadamard and $\Tgate$ gates.} of Hadamard and $\Tgate$ gates using \Cref{lem:single_qubit_T-count}.  To implement $D$, we first use \Cref{lem:boolean_oracle} to compute a Boolean function that takes as input the $n-1$ control bits, and outputs a description of the approximating sequence of Hadamard and $\Tgate$ gates into an ancilla register. Controlled on this ancilla register, we then apply the corresponding sequence of single-qubit gates.

\begin{proof}[Proof of \Cref{thm:diagonal_unitary_T-count}]
Let $D=\diag(\alpha_1,\ldots,\alpha_{2^n})$ where $\alpha_1,\ldots,\alpha_{2^n}$ are complex numbers on the unit circle. Define a diagonal unitary $D'=\diag(\alpha_1,\bar\alpha_1,\ldots,\alpha_{2^n},\bar\alpha_{2^n})$ on $n+1$ qubits, where $\bar{\alpha}_j$ is the complex conjugate of $\alpha_j$. Note that $D'=D\otimes\ketbra00+D^\dag\otimes\ketbra11$. In order to implement $D$, it suffices to implement $D'$ on our $n$-qubit input register along with an ancilla initialized in the $|0\rangle$ state.

For each $j=1,2,\ldots,2^n$, define $U_j=\diag(\alpha_j,\bar\alpha_j)$, which is a single-qubit unitary with determinant $1$. Then $D'=\diag(U_1,\ldots,U_{2^n})$.
Moreover, we know from \Cref{lem:single_qubit_T-count} that each $U_j$ is $\eps$-close to a single-qubit unitary $\tilde U_j$ that can be exactly implemented as a sequence of $O(\log(1/\eps))$ Hadamard and $\Tgate$ gates. It follows that $\|D'-\tilde U\|\leq \eps$, where $\tilde U=\diag(\tilde U_1,\ldots,\tilde U_{2^n})$. To complete the proof we describe a Clifford+$\Tgate$ circuit that exactly implements $\tilde U$ using the number of $\Tgate$ gates and ancillas claimed in the Theorem statement.

From \Cref{lem:single_qubit_T-count} we infer that for some $K=\Theta(\log(1/\eps))$ and each $U_j$, we have
\begin{equation}
\tilde U_j=\Hgate^{a_j^1}\Tgate^{b_j^1}\Hgate^{a_j^2}\Tgate^{b_j^2}\cdots \Hgate^{a_j^K}\Tgate^{b_j^K},
\end{equation}
for some Boolean variables $a_j^1,\ldots,a_j^K,b_j^1,\ldots,b_j^K\in\bin$. Define a Boolean function $f\colon\bin^n\to\bin^{2K}$ to record the gate sequence information $a,b$ for each $\tilde U_j$:
\begin{equation}
f(j)=(a_j^1,b_j^1,\ldots,a_j^K,b_j^K).
\end{equation}
Then by \Cref{lem:boolean_oracle}, the corresponding $(n+2K)$-qubit unitary $F$ mapping $\ket j\ket y$ to $\ket j\ket{y\oplus f(j)}$ can be implemented with $O\pbra{\sqrt{K\cdot2^n}}=O\pbra{\sqrt{2^n\log(1/\eps)}}$ $\Tgate$ gates and ancillas.

Once the gate sequence is computed into an ancilla register, it remains to apply $\Hgate$ or $\Tgate$ on the $(n+1)$-th qubit controlled on this ancilla register. For this purpose, we note that controlled-$\Hgate$ and controlled-$\Tgate$ are two-qubit unitaries that can be exactly implemented with $O(1)$ $\Tgate$ gates without ancillas \cite{giles2013exact}.

Putting this together, we arrive at the following implementation of  $\tilde U$.
\begin{enumerate}
\item\label{itm:lem:block_diag_single_T-count_1} Let $\Acal=(\Acal_0,\Acal_1)$ be the $(n+1)$-qubit register where we implement $\tilde U$. Note that $\Acal_0$ is an $n$-qubit register and $\Acal_1$ is a single-qubit register.

Let $\Bcal=(\Bcal_1,\ldots,\Bcal_{2K})$ and $\Ccal$ be auxiliary registers where each $\Bcal_j$ is single-qubit and $\Ccal$ is $O\pbra{\sqrt{2^n\log(1/\eps)}}$-qubit. These registers are initialized in the all-zero state and serve as ancillas.
\item\label{itm:lem:block_diag_single_T-count_2} Using $\Ccal$, we apply the unitary $F$ described above on $\Acal_0$ and $\Bcal$.
\item\label{itm:lem:block_diag_single_T-count_3} For each $\ell=1,2,\ldots,2K$, we either apply $\Hgate$ on $\Acal_1$ controlled by $\Bcal_\ell$ (if $\ell$ is odd), or we apply  $\Tgate$ on $\Acal_1$ controlled by $\Bcal_\ell$ (if $\ell$ is even).

\item\label{itm:lem:block_diag_single_T-count_4} Using $\Ccal$, we apply the unitary $F$ again to restore the ancillas in $\Bcal$ to the all-zeros state.
\end{enumerate}

To see that this implements $\tilde{U}$ as desired, we can verify the action of each of the above steps when the initial state of registers $\Acal_0$ and $\Acal_1$ are basis vectors $\ket j$ and $|x\rangle$, respectively, where $j\in \{1,2,\ldots, 2^n\}$ and $x\in \{0,1\}$:
\begin{align}
|j\rangle|x\rangle_{\Acal_1}|0\rangle_{\Bcal}|0\rangle_{\Ccal}
&\rightarrow |j\rangle|x\rangle_{\Acal_1}|f(j)\rangle_{\Bcal}|0\rangle_{\Ccal}\\
&\rightarrow |j\rangle \tilde{U}_j|x\rangle_{\Acal_1}|f(j)\rangle_{\Bcal}|0\rangle_{\Ccal}\\
&\rightarrow  |j\rangle \tilde{U}_j|x\rangle_{\Acal_1}|0\rangle_{\Bcal}|0\rangle_{\Ccal}.
\end{align}

The number of $\Tgate$ gates used in the above implementation is upper bounded as
\begin{align}
&\underbrace{O\pbra{\sqrt{2^n\log(1/\eps)}}}_{\text{\Cref{itm:lem:block_diag_single_T-count_2} and \Cref{itm:lem:block_diag_single_T-count_4}}}
+\quad 2K\cdot\underbrace{O(1)}_{\substack{\text{each step} \\ \text{in \Cref{itm:lem:block_diag_single_T-count_3}}}}
=O\pbra{\sqrt{2^n\log(1/\eps)}+\log(1/\eps)}.
\end{align}
It is also clear that this upper bounds the number of ancillas used.

The tightness of \Cref{thm:diagonal_unitary_T-count} is deferred as \Cref{thm:diagonal_adaptive_lb} to be proved in \Cref{sec:lower_bounds}.
\end{proof}

\subsection{Batched synthesis}\label{sec:batch}

We now discuss the applications of \Cref{thm:diagonal_unitary_T-count} mentioned in the introduction. 

First we discuss \textit{batched synthesis} of single-qubit unitaries: we will see that we can implement $m= O( \log\log(1/\eps) )$ single-qubit gates to error $\eps$ with $\Tgate$-count $O(\log(1/\eps))$. 

We are interested in synthesizing a general tensor product $U_1\otimes U_2\otimes \cdots\otimes U_m$ of single-qubit unitaries. In order to use  \Cref{thm:diagonal_unitary_T-count} we will first reduce to the case where each $U_j$ is diagonal in the standard basis. To this end we can use the well-known Euler angle decomposition (see e.g., \cite{nielsen2010quantum}) which expresses a single-qubit unitary $U$ as a product  $U=e^{i\delta} e^{i\gamma\Zgate}e^{i\beta\Xgate}e^{i\alpha\Zgate}$ for some real numbers $\alpha, \beta,\gamma, \delta$. Here $\Xgate,\Ygate,\Zgate$ are the Pauli matrices. Since $e^{i\beta\Xgate}=\Hgate e^{i\beta\Zgate}\Hgate$ we can summarize this fact as follows.

\begin{fact}\label{fct:unitary_syn_n=1}
Let $U$ be a single-qubit unitary.
Then there exist single-qubit diagonal unitaries $A,B,C$ such that $U=A\Hgate B\Hgate C$, where $\Hgate$ is the single-qubit Hadamard matrix.
\end{fact}

Now we reduce the batched synthesis task to \Cref{thm:diagonal_unitary_T-count}.

\begin{corollary}\label{cor:multiple_single_T-count}
Let $m\ge1$ be an integer and $U_1,U_2,\ldots,U_m$ be arbitrary single-qubit unitaries. Define an $m$-qubit unitary $U=U_1\otimes\cdots\otimes U_m$.
Then $U$ can be implemented up to error $\eps$ by a Clifford+$\Tgate$ circuit using
$O\pbra{\sqrt{2^m\log(1/\eps)}+\log(1/\eps)}$ $\Tgate$ gates and ancilla qubits.\footnote{This construction uses $O(2^m\log(1/\eps))$ Clifford gates.}
\end{corollary}
\begin{proof}
By \Cref{fct:unitary_syn_n=1}, each $U_j$ can be expanded as $A_j\Hgate B_j\Hgate C_j$ where $A_j,B_j,C_j$ are single-qubit diagonal unitaries. As a result, we have $U=A\Hgate^{\otimes m}B\Hgate^{\otimes m}C$, where $A=A_1\otimes\cdots\otimes A_m,B=B_1\otimes\cdots\otimes B_m,C=C_1\otimes\cdots\otimes C_m$ are $m$-qubit diagonal unitaries, which can each be implemented using \Cref{thm:diagonal_unitary_T-count} with the claimed number of $\Tgate$ gates and ancillas.
\end{proof}

\Cref{thm:batched_synthesis} then follows directly from \Cref{cor:multiple_single_T-count}.

\batchedsyn*

\begin{proof}
We divide the $m$ single-qubit unitaries into $K=O(1)$ groups of $\ceilbra{\log\log(1/\eps)}$ each.
It then suffices to implement each group up to error $\eps/K$. This is handled by \Cref{cor:multiple_single_T-count}.
\end{proof}

\subsection{Mass production}\label{sec:mass}

Next we consider the \textit{mass production} of a single-qubit unitary. Here the goal is to implement $U^{\otimes m}$ for some single-qubit unitary $U$. 

\massprod*

Before proving this, let us describe our high-level strategy. 
We can use \Cref{fct:unitary_syn_n=1} to reduce to the special case of implementing a diagonal single-qubit unitary $A:=\diag(1,e^{i\theta})$, by choosing the (irrelevant) global phase to be such that all 3 diagonal unitaries are of this form.
For any computational basis input $\ket{x}$ where $x\in\bin^m$, we have $A^{\otimes m}\ket x=e^{i\theta\cdot|x|}|x\rangle $, where $|x|$ is the Hamming weight of $x$. A natural approach is then to first compute the Hamming weight into a second register containing $\lceil \log m \rceil$ qubits, then apply the corresponding phase controlled on this register. The Hamming weight can be computed using $O(m)$ gates, as we discuss below. In the second step we need to implement a diagonal unitary on $\lceil \log m \rceil$ qubits. We can do this using \Cref{thm:diagonal_unitary_T-count}, incurring a $\Tgate$-count $O\pbra{\sqrt{m\log(1/\eps)}+\log(1/\eps)}=O\pbra{m+\log(1/\eps)}$. 

We start by showing how to compute the Hamming weight using $O(m)$ gates.

\begin{fact}\label{fct:hamming_weight_circuit}
Let $U_{\sf Ham}$ be the unitary such that
\begin{equation}
U_\Ham\ket{x}\ket{y}=\ket{x}\ket{y\oplus|x|}
\end{equation}
for all $x\in\bin^m$ and $y\in\bin^{\ceilbra{\log m}}$,
where $|x|\in\cbra{0,1,\ldots,m-1}$ is the Hamming weight of $x$.
Then $U_\Ham$ can be exactly prepared by Clifford+$\Tgate$ circuit using $O(m)$ total gates and ancillary qubits.
\end{fact}

We also remark that \cite{boyar2005exact} gives the exact multiplicative complexity of computing the Hamming weight function, which upper bounds the $\Tgate$-count of implementing $U_{\sf Ham}$ (using the connection in \Cref{rmk:boolean}). Their construction has slightly lower T-count than \Cref{fct:hamming_weight_circuit}.

\begin{proof}
We describe a classical Boolean circuit with $O(m)$ gates to compute the Hamming weight of an $m$-bit string. Here the circuit contains AND, OR, NOT gates with bounded fan-in. Then we can convert this to a reversible classical circuit in the standard way using Toffoli gates and $O(m)$ ancilla qubits; implementing each Toffoli using $O(1)$ $\Tgate$ gates gives \Cref{fct:hamming_weight_circuit}.

Firstly note that given two $t$-bit numbers, we can compute their sum as a $(t+1)$-bit number using $O(t)$ gates, by the textbook algorithm sequentially computing each bit and the corresponding carry bit. The Boolean circuit to compute the Hamming weight will use the above observation iteratively for $t=1,2,\ldots,\log m$ as follows. We view the input as $m$ $1$-bit numbers. Then we partition it into $m/2$ pairs and compute the sum of each pair as a $2$-bit number. Continuing in this way, at the $t$-th stage we will have $m/2^t$  $t$-bit numbers.
The total gate complexity is 
\begin{equation}
\sum_{t=1}^{\ceilbra{\log m}}\frac m{2^t}\cdot O(t)=O(m).
\qedhere
\end{equation}
\end{proof}

Given \Cref{fct:hamming_weight_circuit}, we can now prove \Cref{thm:mass_production}. 

\begin{proof}[Proof of \Cref{thm:mass_production}]
By \Cref{fct:unitary_syn_n=1}, we obtain single-qubit diagonal unitaries $A,B,C$ such that $U=A\Hgate B\Hgate C$.
Thus $U^{\otimes m}=A^{\otimes m}\Hgate^{\otimes m}B^{\otimes m}\Hgate^{\otimes m}C^{\otimes m}$.
Therefore it suffices to show how to implement $A^{\otimes m}$.

By ignoring the (irrelevant) global phase, we assume $A=\diag(1,e^{i\theta})$.
Observe that $A^{\otimes m}\ket{x}=e^{i\cdot\theta\cdot|x|}\ket x$ for all $x\in\bin^m$.
Let $K=\ceilbra{\log(m)}$.
We implement $A^{\otimes m}$ as follows:
\begin{enumerate}
\item Start with $\ket{x}\ket{0^K}$. We apply $U_\Ham$ from \Cref{fct:hamming_weight_circuit} to obtain $\ket{x}\ket{|x|}$.

This uses $O(m)$ $\Tgate$ gates and ancillary qubits by \Cref{fct:hamming_weight_circuit}.
\item Define a $K$-qubit diagonal unitary $D$ such that $D\ket{c}=e^{i\cdot\theta\cdot c}\ket c$. We apply $D$ on the Hamming weight register and obtain $e^{i\cdot\theta\cdot|x|}\ket x\ket{|x|}$.

This uses $O\pbra{\sqrt{m\log(1/\eps)}+\log(1/\eps)}\le O(m+\log(1/\eps))$ $\Tgate$ gates and ancilla qubits by \Cref{thm:diagonal_unitary_T-count}.
\item Finally uncompute the Hamming weight by applying $U_\Ham$ from \Cref{fct:hamming_weight_circuit} again.

This again uses $O(m)$ $\Tgate$ gates and ancilla qubits by \Cref{fct:hamming_weight_circuit}.
\end{enumerate}

Finally we prove the tightness of \Cref{thm:mass_production}. On the one hand, the $\Omega(\log(1/\eps))$ $\Tgate$-count lower bound holds even if $m=1$ \cite{beverland2020lower}.
On the other hand, consider the special case where $U$ is the $\Tgate$ gate. Then $U^{\otimes m}=\Tgate^{\otimes m}$ and it is impossible to implement it using  $o(m)$ $\Tgate$ gates. Otherwise, by gate injection $\ket{\Tgate}^{\otimes m}$ can be approximated with $\Tgate$-count $o(m)$.
Plugging this in \Cref{thm:state_T-count}, we have that any $n$-qubit state can be approximated up to any constant error with $\Tgate$-count $o(\sqrt{2^n})$, contradicting the lower bound for \Cref{thm:state_T-count}.\footnote{The lower bound for \Cref{thm:state_T-count} holds in the presence of Pauli postselections, as needed here after gate injection. See details in \Cref{sec:lower_bounds}.}
\end{proof}

\section{Quantum state preparation with optimal T-count}\label{sec:state_T}

In this section we prove our main result, \Cref{thm:state_T-count} (restated below).

\stateT*

\begin{remark}
Our overall state preparation approach is similar to the one by Rosenthal \cite{rosenthal2024efficient}.
Motivated by the Aaronson-Kuperberg problem \cite{aaronson2007quantum,aaronson2016complexity}, \cite{rosenthal2024efficient} focuses on the number of Boolean oracle calls needed for efficient state preparation.
While the context is different, the $\Tgate$-count of the algorithm presented in \cite{rosenthal2024efficient} is near optimal.
To obtain our \Cref{thm:state_T-count}, we combine Rosenthal's algorithm with the crucial missing pieces \Cref{thm:diagonal_unitary_T-count} and \Cref{thm:batched_synthesis}.
Along the way, we also simplify and improve some of Rosenthal's analysis: our \Cref{lem:apx_state_syn_hd4} simplifies and improves \cite[Lemma 3.2]{rosenthal2024efficient}, which was implicit in \cite{irani2022quantum}; our \Cref{lem:state_circuit} improves the query complexity in \cite[Theorem 4.2]{rosenthal2024efficient}.
\end{remark}

We already know from \Cref{thm:diagonal_unitary_T-count} how to implement an arbitary diagonal unitary using the claimed number of $\Tgate$ gates and ancillas. For the state preparation task, this lets us specialize without loss of generality to the case where the target state has real amplitudes: after preparing a target state with real entries, we can  apply a diagonal unitary to apply any desired complex phases.

As a first step, we show how to synthesize a real-valued target state to within a constant approximation error.  

\begin{lemma}\label{lem:apx_state_syn_hd4}
Let $\ket\psi$ be an arbitrary $n$-qubit state with real amplitudes.
There exists an $n$-qubit state $\ket\phi$ with real amplitudes such that $\braket{\phi}{\psi}\ge\frac1{\sqrt2}$ and $\ket\phi=B_2\Hgate^{\otimes n}B_1\Hgate^{\otimes n}\ket{0^n}$ for some $n$-qubit Boolean phase oracles $B_1$ and $B_2$ (i.e., diagonal unitaries in which all diagonal entries are $\pm 1$).
\end{lemma}

Results similar to \Cref{lem:apx_state_syn_hd4} are obtained in \cite{bravyi2019simulation,irani2022quantum,rosenthal2024efficient}, but our analysis is arguably simpler and achieves better bounds.\footnote{The constant $1/\sqrt{2}$ in this lemma is optimal, as evidenced by the cat state $\ket{\psi}=\frac{1}{\sqrt{2}}(\ket{0^n}+\ket{1^n})$.} The proof of \Cref{lem:apx_state_syn_hd4} is given in \Cref{sec:apx_state_constant}.

To implement the Boolean phase oracles in \Cref{lem:apx_state_syn_hd4}, one can apply \Cref{lem:boolean_oracle} directly, along with the phase kickback trick \cite{cleve1998quantum}.
This is stated as the following \Cref{fct:boolean_phase}.

\begin{fact}\label{fct:boolean_phase}
Let $B$ be an $n$-qubit Boolean phase oracle.
Then $B$ can be prepared exactly by a Clifford+$\Tgate$ circuit using $O\pbra{\sqrt{2^n}}$ $\Tgate$ gates and ancillary qubits.\footnote{This construction uses $O(2^n)$ Clifford gates.}
\end{fact}

As a result of \Cref{fct:boolean_phase}, we can prepare the approximation $\ket\phi$ from \Cref{lem:apx_state_syn_hd4} using only $O\pbra{\sqrt{2^n}}$ $\Tgate$ gates and ancilla qubits. We see that \Cref{lem:apx_state_syn_hd4} already establishes the special case of \Cref{thm:state_T-count} where the error $\eps$ is a sufficiently large constant.

To handle the small error case, we use the fact that \Cref{lem:apx_state_syn_hd4} can be used to approximate \emph{any} quantum state $\ket\psi$. Hence we can reduce the approximation error by approximating the difference state $\ket\psi-\ket\phi$. Continuing iteratively in this way, we are able to reduce the error below any target value, as in \cite{rosenthal2024efficient}. This idea is formalized in \Cref{lem:apx_coeff_simple}, which is proved in \Cref{sec:apx_state_eps}.

\begin{lemma}\label{lem:apx_coeff_simple}
Let $\ket\psi$ be an arbitrary $n$-qubit state with real amplitudes.
For any integer $T\ge1$, there exist $\zeta\in\sbra{\frac{\sqrt2-1}2,\frac1{\sqrt2}}$ and $n$-qubit states $\ket{\psi_0},\ket{\psi_1},\ldots,\ket{\psi_{T-1}}$ such that
\begin{equation}
\vabs{\ket\psi-\zeta\cdot\sum_{k=0}^{T-1}2^{-k/2}\cdot\ket{\psi_k}}\le\frac1{\sqrt2-1}\cdot2^{-T/4},
\end{equation}
where each $\ket{\psi_k}=B_2^{(k)}\Hgate^{\otimes n}B_1^{(k)}\Hgate^{\otimes n}\ket{0^n}$ for some $n$-qubit Boolean phase oracles $B_1^{(k)}$ and $B_2^{(k)}$.
\end{lemma}

Following the strategy from \cite{rosenthal2024efficient}, we then use the linear combination of unitaries (LCU) method \cite{wiebe2012hamiltonian,berry2015hamiltonian} to prepare the weighted sum of states in \Cref{lem:apx_coeff_simple}, flagged by an additional ancilla. We carefully tune the amplitude of the flagged part of the state to enable the use of exact amplitude amplification \cite{grover1998quantum,brassard2002quantum} to get rid of the junk part and obtain the desired $\eps$ approximation. 

Our implementation of this strategy gives the following \Cref{lem:state_circuit}.

\begin{lemma}\label{lem:state_circuit}
Let $\eps\in(0,1/2]$ and $\ket\psi$ be an arbitrary $n$-qubit state with real amplitudes.
There exist some $t=\log\log(1/\eps)+O(1)$ and an $n$-qubit normalized state $\ket\phi$ such that $\vabs{\ket\phi-\ket\psi}\le\eps$ and $\ket{0^t}_\Acal\ket{\phi}_\Bcal=R_1R_2R_1R_2V\ket{0^{n+t}}_{\Acal\Bcal}$ where
\begin{gather}
R_1=V\pbra{2\ketbra{0^{n+t}}{0^{n+t}}-I}V^\dag,\\
R_2=\pbra{2\ketbra{0^t}{0^t}-I}_\Acal\otimes I_\Bcal,
\end{gather}
and $V$ is a quantum circuit with the following structure in order:
\begin{itemize}
\item one layer of single-qubit gates on $\Acal$;
\item Hadamard on $\Bcal$;
\item a Boolean phase oracle on $\Acal\Bcal$;
\item Hadamard on $\Bcal$;
\item another Boolean phase oracle on $\Acal\Bcal$;
\item another layer of single-qubit gates on $\Acal$.
\end{itemize}
\end{lemma}

The proof of \Cref{lem:state_circuit} is deferred to \Cref{sec:flag_and_amplify}. We now prove \Cref{thm:state_T-count}.

\begin{proof}[Proof of \Cref{thm:state_T-count}]
Using \Cref{thm:diagonal_unitary_T-count} to correct phases, we assume without loss of generality that the target state $\ket\psi$ has real amplitudes.
Then it suffices to analyze the $\Tgate$-count and ancilla count for the circuit in \Cref{lem:state_circuit}.
To this end, we analyze the components separately, each of which appears constant number of times in the whole circuit in \Cref{lem:state_circuit}. Assume $\eps\le1/2$.
\begin{itemize}
\item The reflection $2\ketbra{0^{n+t}}{0^{n+t}}-I$ used in $R_1$ is a Boolean phase oracle and can be constructed by \Cref{fct:boolean_phase} with $O\pbra{\sqrt{2^{n+t}}}=O\pbra{\sqrt{2^n\log(1/\eps)}}$ $\Tgate$ gates and ancillas.\footnote{This reflection can be implemented with much lower cost. We use \Cref{lem:boolean_oracle} as a loose upper bound since it does not affect the asymptotic bound in the end.}
\item $R_2$ is the same reflection (with a smaller scale) as above, which is also handled by \Cref{fct:boolean_phase}.
\item The Hadamard layer used in $V$ is Clifford and requires no $\Tgate$ gates or ancilla.
\item The Boolean phase oracle on $\Acal\Bcal$ used in $V$ can be again prepared using \Cref{fct:boolean_phase} with $O\pbra{\sqrt{2^{n+t}}}=O\pbra{\sqrt{2^n\log(1/\eps)}}$ $\Tgate$ gates and ancillas.
\item The layer of single-qubit gates on $\Acal$ used in $V$ can be synthesized by \Cref{thm:batched_synthesis} since $\Acal$ only has $t=\log\log(1/\eps)+O(1)$ qubits. This uses $O\pbra{\log(1/\eps)}$ $\Tgate$ gates and ancillas by \Cref{thm:batched_synthesis}.
The approximation error does not blow up since this is only used a constant number of times in total.
\end{itemize}
Summing over costs of the above components establishes \Cref{thm:state_T-count}.
The tightness of \Cref{thm:state_T-count} is deferred as \Cref{thm:state_adaptive_lb} to be proved in \Cref{sec:lower_bounds}.
\end{proof}

\subsection{Approximating a state to constant error: \texorpdfstring{\Cref{lem:apx_state_syn_hd4}}{Lemma \ref*{lem:apx_state_syn_hd4}}}\label{sec:apx_state_constant}

In this section we describe how to approximate a state to constant error (\Cref{lem:apx_state_syn_hd4}) and in the next section we discuss the finer approximation (\Cref{lem:apx_coeff_simple}).

To get the coarse approximation of \Cref{lem:apx_state_syn_hd4}, we take our target state and apply a random diagonal unitary with $\pm 1$ entries followed by a binary Fourier transform (i.e., the Hadamard matrix).  This does a reasonable job of flattening the amplitudes of the state, by which we mean that most of the amplitudes are now of similar magnitude. We then apply a diagonal unitary to make all the phases $+1$, which gives a state that has constant overlap with the uniform superposition $|+^n\rangle$. To implement this strategy we use the standard Khintchine inequality.

\begin{fact}[Khintchine inequality \cite{haagerup1981best,wiki:Khintchine_inequality}] \label{fct:khintchine}
Let $N\ge1$ be an integer and $\beta_1,\ldots,\beta_N$ be arbitrary complex numbers satisfying $\sum_{j\in[N]}|\beta_j|^2=1$.
Then
\begin{equation}
\frac1{\sqrt2}\le\E_{x\sim\binpm^N}\bigg[\bigg|\sum_{j\in[N]}\beta_j\cdot x_j\bigg|\bigg]\le1.
\end{equation}
\end{fact}

Using \Cref{fct:khintchine}, we show how to flatten the amplitudes of a quantum state.

\begin{lemma}\label{lem:flattening_dh2}
Let $\ket\psi$ be an arbitrary $n$-qubit state. Let $F= H^{\otimes n}$ or more generally, let $F$ be any $n$-qubit unitary such that $|F_{ij}|=1/\sqrt{2^n}$ for $i,j\in \bin^n$. Then
\begin{equation}
\E_{x\sim\binpm^{\bin^n}}\sbra{\vabs{F ~ \diag(x)\ket\psi}_1}\ge\frac{\sqrt{2^n}}{\sqrt2}.
\end{equation}
\end{lemma}
\begin{proof} 
Write $\ket\psi=\sum_{j\in\bin^n}\beta_j\ket j$ with $\sum_{j\in \bin^n}|\beta_j|^2=1$. Then 
\begin{align}
\E_{x}\sbra{\vabs{F ~ \diag(x)\ket\psi}_1}
&= \E_{x}\biggl[\biggl\|\sum_{i,j}F_{ij}x_j\beta_j\ket i\biggr\|_1 \biggr]\\
&= \E_{x}\biggl[ \sum_{i} \biggl|\sum_{j} F_{ij}x_j\beta_j \biggr| \biggr]
= \sum_{i} \E_{x}\biggl[\biggl|\sum_{j} F_{ij}x_j\beta_j \biggr| \biggr],
\end{align}
where the last equality used the linearity of expectation. By assumption, 
$F_{ij}=f_{ij}/\sqrt{2^n}$, where $f_{ij}$ is a complex number of unit modulus. 
Then $f_{ij}\beta_j$ satisfies the assumption of \Cref{fct:khintchine} since $\sum_{j\in \bin^n}|f_{ij}\beta_j|^2=1$, and hence 
\begin{align}
\E_{x}\sbra{\vabs{F ~ \diag(x)\ket\psi}_1}
&= \frac{1}{\sqrt{2^n}} \sum_{i} \E_{x} \biggl[\biggl|\sum_{j} f_{ij}x_j\beta_j \biggr| \biggr]
\geq \frac{1}{\sqrt{2}} \sum_{i} \frac{1}{\sqrt{2^n}}
= \frac{\sqrt{2^n}}{\sqrt2}
\end{align}
as claimed.
\end{proof}

After flattening, the state becomes essentially a uniform superposition with different phases. 
Moreover, in the case of interest where $\ket\psi$ is assumed to have real amplitudes, these phases are simply $\pm1$.
Then we can correct the Boolean phases and reverse the flattening procedure to obtain an approximation of the original state $\ket\psi$.

\begin{proof}[Proof of \Cref{lem:apx_state_syn_hd4}]
Let $\ket\psi$ be an $n$-qubit state with real amplitudes. By \Cref{lem:flattening_dh2}, there exists a Boolean phase oracle $B_2$ such that $|\tilde\psi\rangle:=\Hgate^{\otimes n}B_2\ket\psi$ has $\ell_1$ norm 
\begin{equation}
\|| {\tilde\psi} \rangle\|_1\ge\frac{\sqrt{2^n}}{\sqrt2}.
\end{equation}
Since $\ket\psi$ has real amplitudes, so does $|\tilde{\psi}\rangle$.
Let $B_1$ be the Boolean phase oracle such that $\langle x|B_1|x\rangle=\mathrm{sign}(\langle x|\tilde{\psi}\rangle)$ for all $x\in \{0,1\}^n$.

Letting $\ket\phi=B_2\Hgate^{\otimes n}B_1\Hgate^{\otimes n}\ket{0^n}$, we have
\begin{align}
\braket{\psi}{\phi}
&=\bra\psi B_2\Hgate^{\otimes n}B_1\Hgate^{\otimes n}\ket0
=\langle\tilde{\psi}|B_1\Hgate^{\otimes n}\ket0
=\frac1{\sqrt{2^n}}\cdot\||\tilde\psi\rangle\|_1
\ge\frac1{\sqrt2},
\end{align}
as desired.
\end{proof}

\subsection{Approximating a state to \texorpdfstring{$\eps$}{eps} error: \texorpdfstring{\Cref{lem:apx_coeff_simple}}{Lemma \ref*{lem:apx_coeff_simple}}}\label{sec:apx_state_eps}

Now we discuss how to reduce the approximation error. We prove the following statement and show that  \Cref{lem:apx_coeff_simple} follows from it. 
The proof follows the reasoning from \cite[Lemma 3.1]{rosenthal2024efficient}, but completes an edge case (see ``Case $T>j^*$'' in the proof) missed in the previous analysis.

\begin{lemma}\label{lem:apx_coeff}
Let $\alpha\in(0,1]$ be a real number.
For any $n$-qubit state $\ket\phi$, let $U_\phi$ be an $n$-qubit unitary that satisfies $\Re\bra\phi U_\phi\ket{0^n}\ge\alpha$.\footnote{For a complex number $v$, $\Re v$ equals its real part.}
Define 
\begin{equation}
\gamma=\min\cbra{\alpha,\frac1{\sqrt2}}
\quad\text{and}\quad
\beta=\sqrt{1-\gamma^2}.
\end{equation}

Then for every $n$-qubit state $\ket\psi$ and integer $T\ge1$, there exist $\gamma\cdot(1-\beta)\le\zeta\le\gamma$ and $\ket{\psi_0},\ket{\psi_1},\ldots,\ket{\psi_{T-1}}$ such that
\begin{align}\label{eq:lem:apx_coeff_1}
\vabs{\ket\psi-\zeta\cdot\sum_{k=0}^{T-1}\beta^k\cdot U_{\psi_k}\ket{0^n}}\le\frac\gamma{1-\beta}\cdot\beta^{T/2}.
\end{align}
\end{lemma}
\begin{proof}
Since $\Re\bra\phi U_\phi\ket{0^n}\ge\alpha\ge\gamma$, we have
\begin{align}\label{eq:lem:apx_coeff_2}
\vabs{\ket\phi-\gamma\cdot U_\phi\ket{0^n}}
&=\sqrt{1+\gamma^2-2\gamma\cdot\Re\bra\phi U_\phi\ket{0^n}}
\le\sqrt{1-\gamma^2}
=\beta.
\end{align}
Since $\gamma>0$, we know $\beta<1$.
Thus our task of approximating $\ket\psi$ is now reduced to that of approximating $\ket\psi-\gamma\cdot U_\psi\ket{0^n}$ which has a smaller length.
We will show that \Cref{eq:lem:apx_coeff_1} follows by iteratively applying \Cref{eq:lem:apx_coeff_2}.

Formally, define $\ket{\psi_0}=|\tilde\psi_0\rangle=\ket\psi$ and recursively for $j\ge1$
\begin{equation}
\ket{\tilde\psi_j}=\ket\psi-\gamma\cdot\sum_{k=0}^{j-1}\beta^k\cdot U_{\psi_k}\ket{0^n}
\end{equation}
and
\begin{equation}
\ket{\psi_j}=\frac{\ket{\tilde\psi_j}}{\vabs{\ket{\tilde\psi_j}}}
\end{equation}
until some $j^*\ge1$ for which $\ket{\tilde\psi_{j^*}}$ becomes a zero vector.
Set $j^*=+\infty$ if $\ket{\tilde\psi_j}$ is never zero.

\paragraph{Case $T\le j^*$.}
We first handle the case where $T\le j^*$.
Below we give an inductive proof that
\begin{align}\label{eq:lem:apx_coeff_3}
\vabs{\ket{\tilde\psi_j}}\le\beta^j
\end{align}
for each $1\le j\le j^*$. 
Then we simply set $\zeta=\gamma$ and \Cref{eq:lem:apx_coeff_1} follows from \Cref{eq:lem:apx_coeff_3} as $\gamma,\beta\in[0,1]$.

To establish \Cref{eq:lem:apx_coeff_3}, first note that the base case $j=1$ follows directly from \Cref{eq:lem:apx_coeff_2}. For $j\ge2$, we have
\begin{align}
\vabs{\ket{\tilde\psi_j}}^2
&=\vabs{\ket{\tilde\psi_{j-1}}-\gamma\beta^{j-1}\cdot U_{\psi_{j-1}}\ket{0^n}}^2\\
&=\vabs{\vabs{\ket{\tilde\psi_{j-1}}}\cdot\ket{\psi_{j-1}}-\gamma\beta^{j-1}\cdot U_{\psi_{j-1}}\ket{0^n}}^2\\
&=\vabs{\ket{\tilde\psi_{j-1}}}^2+\gamma^2\beta^{2j-2}
-2\gamma\beta^{j-1}\vabs{\ket{\tilde\psi_{j-1}}}\cdot\Re\bra{\psi_{j-1}}U_{\psi_{j-1}}\ket{0^n}\\
&\le\vabs{\ket{\tilde\psi_{j-1}}}^2+\gamma^2\beta^{2j-2}-2\gamma^2\beta^{j-1}\vabs{\ket{\tilde\psi_{j-1}}}
\tag{since $\Re\bra{\psi_{j-1}}U_{\psi_{j-1}}\ket{0^n}\ge\gamma$}.
\end{align}
Since $\gamma\le\frac1{\sqrt2}$ and $\beta=\sqrt{1-\gamma^2}$, we know $\gamma\le\beta$. 
Then by the induction hypothesis we have
\begin{align}
\vabs{\ket{\tilde\psi_j}}^2
&\le\max_{0\le x\le\beta^{j-1}}x^2+\gamma^2\beta^{2j-2}-2\gamma^2\beta^{j-1}x\\
&=\max\cbra{\gamma^2\beta^{2j-2},\beta^{2j-2}-\gamma^2\beta^{2j-2}}\\
&=\max\cbra{\gamma^2\beta^{2j-2},\beta^{2j}}
=\beta^{2j}
\tag{since $\gamma^2+\beta^2=1$ and $\gamma\le\beta$}
\end{align}
as desired.

\paragraph{Case $T>j^*$.}
Now we turn to the case where $T>j^*$.
For each $j\ge j^*\ge1$, define $\ket{\psi_j}=\ket{\psi_{j\mod j^*}}$.
Let $t=\floorbra{\frac T{j^*}}$ and $m=T\mod j^*$.
Then
\begin{align}
\sum_{k=0}^{T-1}\beta^k\cdot U_{\psi_k}\ket{0^n}
&=\sum_{\ell=0}^{t-1}\beta^{\ell\cdot j^*}\sum_{k=0}^{j^*-1}\beta^k\cdot U_{\psi_k}\ket{0^n}
+\sum_{k=0}^m\beta^{t\cdot j^*+k}\cdot U_{\psi_k}\ket{0^n}\\
&=\sum_{\ell=0}^{t-1}\frac{\beta^{\ell\cdot j^*}}{\gamma}\cdot\ket\psi+\sum_{k=0}^m\beta^{t\cdot j^*+k}\cdot U_{\psi_k}\ket{0^n}
\tag{by the definition of $j^*$}\\
&=\frac{1-\beta^{t\cdot j^*}}{\gamma\cdot(1-\beta^{j^*})}\cdot\ket\psi+\sum_{k=0}^m\beta^{t\cdot j^*+k}\cdot U_{\psi_k}\ket{0^n}.
\end{align}
Define $\zeta=\frac{\gamma\cdot(1-\beta^{j^*})}{1-\beta^{t\cdot j^*}}$.
Then we have 
\begin{equation}
\gamma\ge\zeta\ge\gamma\cdot(1-\beta^{j^*})\ge\gamma\cdot(1-\beta)
\end{equation}
as claimed.
In addition,
\begin{align}
\text{LHS of \Cref{eq:lem:apx_coeff_1}}
&=\vabs{\zeta\cdot\sum_{k=0}^m\beta^{t\cdot j^*+k}\cdot U_{\psi_k}\ket{0^n}}
\le\zeta\cdot\sum_{k=0}^m\beta^{t\cdot j^*+k}\\
&\le\zeta\cdot\sum_{k=0}^{j^*-1}\beta^{t\cdot j^*+k}
=\zeta\cdot\frac{1-\beta^{j^*}}{1-\beta}\cdot\beta^{t\cdot j^*}
\tag{since $m=T\mod j^*\le j^*-1$}\\
&\le\gamma\cdot\frac{1-\beta^{j^*}}{1-\beta}\cdot\beta^{t\cdot j^*}
\le\frac\gamma{1-\beta}\cdot\beta^{t\cdot j^*}
\tag{since $\zeta\le\gamma$}\\
&=\frac\gamma{1-\beta}\cdot\beta^{\floorbra{\frac T{j^*}}\cdot j^*}
\tag{since $t=\floorbra{\frac T{j^*}}$}\\
&\le\frac\gamma{1-\beta}\cdot\beta^{T/2},
\tag{since $T\ge j^*$ and $\floorbra{x}\ge x/2$ for $x\ge1$}
\end{align}
which verifies \Cref{eq:lem:apx_coeff_1}.
\end{proof}

Given \Cref{lem:apx_coeff}, we immediately obtain \Cref{lem:apx_coeff_simple}.

\begin{proof}[Proof of \Cref{lem:apx_coeff_simple}]
In light of \Cref{lem:apx_coeff}, we set $\alpha=\frac1{\sqrt2}$ and design $U_\phi$ as the $B\Hgate B\Hgate$ circuit from \Cref{lem:apx_state_syn_hd4}.
Note that all states used in \Cref{lem:apx_coeff} have real amplitudes since we start with $\ket\psi$ that has real amplitudes.
Therefore circuits $U_\phi$ are all well-defined.
\end{proof}

\subsection{Flagging and exact amplitude amplification: \texorpdfstring{\Cref{lem:state_circuit}}{Lemma \ref*{lem:state_circuit}}}\label{sec:flag_and_amplify}

In this Section we describe the circuit that prepares an $n$-qubit state with optimal $\Tgate$-count, establishing \Cref{lem:state_circuit}.

We first use the linear combination of unitaries technique to deduce a circuit that prepares the approximating state from \Cref{lem:apx_coeff_simple}.

\begin{lemma}\label{lem:coarse_apx_flag}
Let $\eps\in(0,1/2]$ and $\ket\psi$ be an arbitrary $n$-qubit state with real amplitudes.
There exist some $t=\log\log(1/\eps)+O(1)$, an $n$-qubit normalized state $\ket\phi$, and an $(n+t)$-qubit unnormalized state $\ket\tau$ such that 
\begin{enumerate}
\item\label{itm:lem:coarse_apx_flag_1}
$\vabs{\ket\phi-\ket\psi}\le\eps$.
\item\label{itm:lem:coarse_apx_flag_2}
$\pbra{\ketbra{0^t}{0^t}\otimes I_n}\ket\tau=0$, where $I_n$ is the $n$-qubit identity operator.
\item\label{itm:lem:coarse_apx_flag_3}
For some $\xi\ge\frac{4\sqrt2-4}5\ge0.33$, the normalized state $\xi\ket{0^t}_\Acal\ket{\phi}_\Bcal+\ket{\tau}_{\Acal\Bcal}$ can be prepared as follows, starting from the initial state $\ket{0^t}_\Acal\ket{0^n}_\Bcal$:
\begin{enumerate}
\item\label{itm:lem:coarse_apx_flag_3a} Apply one layer of single-qubit gates on $\Acal$.
\item\label{itm:lem:coarse_apx_flag_3b} Apply Hadamard on $\Bcal$.
\item\label{itm:lem:coarse_apx_flag_3c} Apply a Boolean phase oracle on $\Acal\Bcal$.
\item\label{itm:lem:coarse_apx_flag_3d} Apply Hadamard on $\Bcal$.
\item\label{itm:lem:coarse_apx_flag_3e} Apply another Boolean phase oracle on $\Acal\Bcal$.
\item\label{itm:lem:coarse_apx_flag_3f} Apply another layer of single-qubit gates on $\Acal$.
\end{enumerate}
\end{enumerate}
\end{lemma}
\begin{proof}
Let $T=2^t$ and $\beta=\frac1{\sqrt2}$.
By \Cref{lem:apx_coeff_simple}, there exist states $\ket{\psi_0},\ket{\psi_1},\ldots,\ket{\psi_{T-1}}$ and Boolean phase oracles $B_1^{(k)},B_2^{(k)}$ for $0\leq k\leq T-1$, such that $\ket{\psi_k}=B_2^{(k)}\Hgate^{\otimes n}B_1^{(k)}\Hgate^{\otimes n}\ket{0^n}$ and 
\begin{align}\label{eq:lem:coarse_apx_flag_1}
\vabs{\ket\psi-\zeta\cdot\sum_{k=0}^{T-1}\beta^k\ket{\psi_k}}\le\frac1{\sqrt2-1}\cdot2^{-T/4}\le\frac\eps{10},
\end{align}
where $\zeta\in\sbra{\frac{\sqrt2-1}2,\frac1{\sqrt2}}$ and for the last inequality we used the fact that $t=\log\log(1/\eps)+O(1)$.

Define $\ell=\vabs{\zeta\cdot\sum_{k=0}^{T-1}\beta^k\ket{\psi_k}}$ and
\begin{align}\label{eq:lem:coarse_apx_flag_2}
\ket\phi=\frac{\zeta}\ell\cdot\sum_{k=0}^{T-1}\beta^k\ket{\psi_k}
\quad\text{and}\quad
\xi=\frac{\ell\cdot(1-\beta)}{\zeta\cdot(1-\beta^T)}.
\end{align}
To get the circuit for \Cref{itm:lem:coarse_apx_flag_3} we use the linear combination of unitaries technique:
\begin{itemize}
\item In \Cref{itm:lem:coarse_apx_flag_3a} we apply a layer of single-qubit gates on $\Acal$ that transforms $\ket{0^t}_\Acal$ to 
\begin{align}
\sqrt{\frac{1-\beta}{1-\beta^T}}\prod_{\ell=0}^{t-1}\pbra{\ket0+\beta^{2^\ell/2}\ket1}_\Acal
&=\sqrt{\frac{1-\beta}{1-\beta^T}}\sum_{k=0}^{T-1}\beta^{k/2}\ket{k}_\Acal.
\end{align}
\item \Cref{itm:lem:coarse_apx_flag_3b,itm:lem:coarse_apx_flag_3c,itm:lem:coarse_apx_flag_3d,itm:lem:coarse_apx_flag_3e} correspond to, controlled on index $k$ at register $\Acal$, $B_2^{(k)}\Hgate^{\otimes n}B_1^{(k)}\Hgate^{\otimes n}$ is applied on $\ket{0^n}_\Bcal$ to obtain
\begin{equation}
\sqrt{\frac{1-\beta}{1-\beta^T}}\sum_{k=0}^{T-1}\beta^{k/2}\ket{k}_\Acal\otimes\ket{\psi_k}_\Bcal.
\end{equation}
Here we use the fact that, since each $B_1^{(k)},B_2^{(k)}$ are Boolean phase oracles, their controlled versions are Boolean phase oracles on registers $\Acal\Bcal$.
\item In \Cref{itm:lem:coarse_apx_flag_3f} we apply the inverse of the layer of single-qubit gates from \Cref{itm:lem:coarse_apx_flag_3a}, which gives 
\begin{align}
\ket{0^t}_\Acal\cdot\frac{1-\beta}{1-\beta^T}\sum_{k=0}^{T-1}\beta^k\ket{\psi_k}_\Bcal+\ket{\tau}_{\Acal\Bcal}
&=\xi\ket{0^t}_\Acal\ket{\phi}_\Bcal+\ket{\tau}_{\Acal\Bcal}
\end{align}
by \Cref{eq:lem:coarse_apx_flag_2}.
From the above we see that \Cref{itm:lem:coarse_apx_flag_2} is satisfied.
\end{itemize}

It remains to prove \Cref{itm:lem:coarse_apx_flag_1} and verify the range of $\xi$ in \Cref{itm:lem:coarse_apx_flag_3}.
By \Cref{eq:lem:coarse_apx_flag_1}, we have $\vabs{\ket\psi-\ell\cdot\ket\phi}\le\frac{\beta^{T/2}}{\sqrt2-1}$ and hence 
\begin{align}\label{eq:lem:coarse_apx_flag_3}
1-\frac{\beta^{T/2}}{\sqrt2-1}\le\ell\le1+\frac{\beta^{T/2}}{\sqrt2-1}.
\end{align}
Therefore $\vabs{\ket\phi-\ket\psi}\le\vabs{\ket\psi-\ell\cdot\ket\phi}+\abs{1-\ell}\le\frac{2\cdot\beta^{T/2}}{\sqrt2-1}\le\eps$, which verifies \Cref{itm:lem:coarse_apx_flag_1}.
Plugging \Cref{eq:lem:coarse_apx_flag_1} and \Cref{eq:lem:coarse_apx_flag_3} into \Cref{eq:lem:coarse_apx_flag_2}, we have
\begin{align}
\xi
&\ge\frac{\pbra{1-\beta^{T/2}}\cdot(1-\beta)}{\zeta\cdot(1-\beta^T)}
=\frac{1-\beta}{\zeta\cdot\pbra{1+\beta^{T/2}}}
\ge\frac{1-\beta}{\zeta\cdot(1+\frac\eps2)}
\tag{since $\beta^{T/2}\le\frac\eps2$}\\
&\ge\frac{1-\beta}{\frac1{\sqrt2}\cdot\frac54}
\tag{since $\eps\le1/2$ and $\zeta\le\frac1{\sqrt2}$}\\
&=\frac{4\cdot(\sqrt2-1)}5\ge0.33
\tag{since $\beta=\frac1{\sqrt2}$}
\end{align}
as claimed.
\end{proof}

\Cref{itm:lem:coarse_apx_flag_3} in \Cref{lem:coarse_apx_flag} is a coarse flagging scheme where the approximation $\ket\phi$ of the target $\ket\psi$ is flagged with some undetermined amplitude $\xi\geq 0.33$. In order to enable the use of exact amplitude amplification, it will be convenient to fix this amplitude to a known value.

\begin{corollary}\label{cor:exact_apx_flag}
\Cref{lem:coarse_apx_flag} holds with the replacement $\xi=\sin\pbra{\frac\pi{10}}$.
\end{corollary}
\begin{proof}
By \Cref{lem:coarse_apx_flag}, $\xi\ge0.33\ge0.31\ge\sin\pbra{\frac\pi{10}}$.
Therefore there exists a single-qubit unitary $G$ such that
\begin{equation}
G\ket0=\frac{\sin(\pi/10)}\xi\cdot\ket0+\sqrt{1-\frac{\sin^2(\pi/10)}{\xi^2}}\cdot\ket1.
\end{equation}
Recall that \Cref{lem:coarse_apx_flag} constructs the state $\xi\ket{0^t}_\Acal\ket{\phi}_\Bcal+\ket{\tau}_{\Acal\Bcal}$.
Now add an extra ancilla in $\Acal$ and apply $G$ to it. 
This gives the state $\sin\pbra{\frac\pi{10}}\cdot\ket{0^{t+1}}_{\Acal}\ket{\phi}_\Bcal+\ket{\tau'}_{\Acal\Bcal}$, where $(\ketbra{0^{t+1}}{0^{t+1}}\otimes I_n)\ket{\tau'}=0$.
Finally, note that the additional gate $G$ can be absorbed into the layer of single-qubit gates applied in \Cref{itm:lem:coarse_apx_flag_3f} of the circuit described in \Cref{lem:coarse_apx_flag}.
We have shown that we can safely set $\xi=\sin\pbra{\frac\pi{10}}$ in \Cref{lem:coarse_apx_flag}.
\end{proof}

The constant $\sin\pbra{\frac\pi{10}}$ is chosen because it allows us to prepare $\ket\phi$ exactly after two rounds of amplitude amplification \cite{grover1998quantum,brassard2002quantum}.
We now complete the proof of \Cref{lem:state_circuit}.

\begin{proof}[Proof of \Cref{lem:state_circuit}]
Let $V$ be the circuit in \Cref{cor:exact_apx_flag} to exactly prepare
\begin{equation}
\ket\rho:=\sin\pbra{\frac\pi{10}}\cdot\ket{0^t}_{\Acal}\ket{\phi}_\Bcal+\ket{\tau}_{\Acal\Bcal}.
\end{equation}
By \Cref{lem:coarse_apx_flag} and \Cref{cor:exact_apx_flag}, $V$ has the structure claimed in \Cref{lem:state_circuit}.

Observe that $R_1=V\pbra{2\ketbra{0^{n+t}}{0^{n+t}}-I}V^\dag$ is the reflection along the state $\ket\rho$.
Also recall that $R_2=(2\ketbra{0^t}{0^t}-I)_{\Ccal\Acal}\otimes I_\Bcal$ is the reflection on register $\Ccal\Acal$ along $\ket{0^t}$.
Using the standard analysis of amplitude amplification \cite{grover1998quantum,brassard2002quantum} and the fact that $\sin(5\cdot\frac\pi{10})=1$, we arrive at
\begin{equation}
R_1R_2R_1R_2V\ket{0^{n+t}}=(-R_1R_2)^2\ket\rho=\ket{0^t}\ket\phi,
\end{equation}
as desired.
\end{proof}

\section{Lower bounds}\label{sec:lower_bounds}

In this section we prove $\Tgate$-count lower bounds matching the upper bounds in \Cref{thm:state_T-count} and \Cref{thm:diagonal_unitary_T-count}. Our lower bounds apply in a very general model of adaptive Clifford+$\Tgate$ circuits that we now describe. 

\paragraph*{Adaptive Clifford+$\Tgate$ Circuits.}
 Consider a quantum computation $\Acal$ with an $n$-qubit input register as well as an $a$-qubit ancilla register that is initialized in the all-zero state. 
The computation proceeds via a sequence of Clifford gates, $\Tgate$ gates, and single-qubit measurements in the computational basis. 
These operations may act on any of the $n+a$ qubits. Moreover, each Clifford or $\Tgate$ gate may be classically controlled on the outcomes of measurements that have occurred previously. 

Let $r\in \{0,1\}^*$ be a binary string that describes the measurement outcomes in order obtained during the course of the computation.\footnote{For maximum generality, we do not assume that the number of single-qubit measurements is bounded. If we knew there were exactly $m$ measurements, then $r$ would be an $m$-bit string.} For each $r$, there is a Kraus operator $K_r$ which describes the corresponding sequence of $\Tgate$ gates, Clifford gates, and single-qubit projectors $|r_i\rangle\langle r_i|$ for $i=1,2,\ldots$. The adaptive Clifford+$\Tgate$ circuit implements a quantum channel which maps the $n$-qubit input state $\rho$ to an $(n+a)$-qubit output state:
\begin{equation}\label{eq:kraus_1}
\Acal(\rho)= \sum_r K_r \left(\rho\otimes |0^a\rangle \langle 0^a| \right)K_r^{\dagger}.
\end{equation}
For example, in the simplest case where there are no measurements and the circuit simply applies an $(n+a)$-qubit Clifford+$\Tgate$ unitary $U$, $\Acal(\rho)= U\left(\rho\otimes |0^a\rangle \langle 0^a| \right)U^{\dagger}$. As another example, if the circuit applies a unitary $U$ and then measures the first qubit, then we have $\Acal(\rho)= \sum_{r \in \{0,1\}} (|r\rangle\langle r| \otimes I)U \left(\rho\otimes |0^a\rangle \langle 0^a| \right)U^{\dagger}(|r\rangle\langle r| \otimes I)$.

We can view the output state $\Acal(\rho)$ as a probabilistic mixture of the states
\begin{align}\label{eq:kraus_2}
\sigma_r(\rho)\equiv \frac{K_r \left(\rho\otimes |0^a\rangle \langle 0^a| \right) K_r^{\dagger}}{\tr(K_r^{\dagger} K_r \left(\rho\otimes |0^a\rangle \langle 0^a| \right))},
\end{align}
where $\sigma_r(\rho)$ occurs with probability $p_r(\rho)\equiv \tr(K_r^{\dagger} K_r \left(\rho\otimes |0^a\rangle \langle 0^a| \right))$.

\paragraph*{T-Count.}
Let $\Tcal_r$ be the number of $\Tgate$ gates used by the circuit, conditioned on measurement outcome $r$. 
The maximal number of $\Tgate$ gates used by $\Acal$ is
\begin{align}
\Tcal_{\mathrm{max}}(\Acal)=\sup_{r\in \bin^*} \Tcal_r.
\end{align}
We also consider the expected $\Tgate$-count of $\mathcal{A}$ starting with the all-zeros input state:
\begin{align}
\Tcal^{0}(\Acal)=\sum_r p_r(|0^n\rangle\langle 0^n|) \Tcal_r. 
\end{align}
Finally, we define the expected $\Tgate$-count of $\Acal$ for the worst-case input state as follows
\begin{align}
\Tcal(\Acal)=\sup_{\rho} \sum_r p_r(\rho) \Tcal_r, 
\end{align}
where $\rho$ is any $n$-qubit mixed state and $\sum_r p_r(\rho) \Tcal_r$ is the expected number of $\Tgate$ gates used in $\Acal(\rho)$. For example, if the circuit is unitary and has no measurements, then all three of these measures equal the $\Tgate$-count of the unitary.

Note that clearly we have
\begin{align}\label{eq:t_complexity_1}
\Tcal^0(I\otimes\Acal)
= \Tcal^0(\Acal)
\leq \Tcal(I\otimes\Acal)
= \Tcal(\Acal)\leq \Tcal_{\mathrm{max}}(\Acal)
=\Tcal_{\mathrm{max}}(I\otimes\Acal).
\end{align}

Here we used the fact that, if $\mathcal{A}$ is an adaptive Clifford+$\Tgate$ circuit with an $n$-qubit input register, then so is $\mathcal{A}'=I\otimes \mathcal{A}$ where the first register can have any number of qubits $n'\geq 0$ and the second register is $n$ qubits (in particular, $\mathcal{A}'$ is just $\mathcal{A}$ applied to the second register).

In the following we will prove lower bounds on the expected $\Tgate$-count $\mathcal{T}^0(\mathcal{A})$ (for state preparation) and $\mathcal{T}(\mathcal{A})$ (for unitary synthesis). These in turn imply lower bounds on the maximal $\Tgate$-count due to 
 \Cref{eq:t_complexity_1}.

\paragraph*{State Preparation.}
We say that $\Acal$ prepares state $\ket\psi$ to within error $\eps$ in trace distance if
\begin{align}
\label{eq:trdist}
\frac12\cdot\|\Acal(\ketbra{0^n}{0^n})-|\psi\rangle\langle \psi|\otimes |0^a\rangle\langle 0^a|\|_1\leq \eps,
\end{align}
where $\|\cdot\|_1$ here is the trace norm. 

Note that one could instead consider a variant of the approximate state preparation task where we do not require that the ancillas are returned to the all-zeros state. In this case we would require the following closeness condition, which might seem to be less stringent than \Cref{eq:trdist}:
\begin{equation}
\label{eq:morestringent}
\frac{1}{2}\cdot\left\|\mathrm{Tr}_{\mathrm{anc}}(\mathcal{A}(|0^n\rangle\langle 0^n|))-|\psi\rangle\langle \psi|\right\|_1 \leq \eps,
\end{equation}
where $\mathrm{Tr}_{\mathrm{anc}}$ denotes the partial trace over the $a$-qubit ancilla register. However it is not hard to see that we can work with the closeness condition \Cref{eq:trdist} without loss of generality. Indeed, if an adaptive Clifford+$\Tgate$ circuit $\mathcal{A}$ satisfies \Cref{eq:morestringent}, then we can construct another adaptive Clifford+$\Tgate$ circuit $\mathcal{B}$ that satisfies \Cref{eq:trdist} and with the same expected $\Tgate$-count  $\mathcal{T}^0(\mathcal{A})=\mathcal{T}^0(\mathcal{B})$. In particular  $\mathcal{B}$ is obtained by first applying $\mathcal{A}$, then measuring the ancilla register in the computational basis, then applying a Pauli gate on the ancilla register which is controlled on the measurement outcome $z\in \{0,1\}^a$, that maps $|z\rangle\rightarrow |0^a\rangle$. Then $\mathcal{B}(|0^n\rangle \langle 0^n|)=\mathrm{Tr}_{\mathrm{anc}}(\mathcal{A}(|0^n\rangle\langle 0^n|))\otimes |0^a\rangle\langle 0^a|$ and the closeness condition \Cref{eq:trdist} holds for $\mathcal{B}$.

The following \Cref{thm:state_adaptive_lb} is a strengthening of the tightness of \Cref{thm:state_T-count}, which we will prove in this section.

\begin{theorem}\label{thm:state_adaptive_lb}
There is an $n$-qubit state $|\psi\rangle$ such that any adaptive Clifford+$\Tgate$ circuit $\Acal$ that prepares $\ket\psi$ up to error $\eps$ in trace distance satisfies  
$$
\Tcal^0(\Acal)=\Omega\pbra{\sqrt{2^n\log(1/\eps)}+\log(1/\eps)}.
$$
\end{theorem}

This theorem lower bounds the expected $\Tgate$-count $\mathcal{T}^0(\mathcal{A})$ of an adaptive Clifford+$\Tgate$ circuit--this is the most general setting we are able to handle. Note that if $\mathcal{A}$ is in fact a unitary $U$ composed of Clifford gates and $\Tgate$ gates (with no measurement), then its output state is a pure state $\mathcal{A}(\rho)= U(\rho\otimes |0^a\rangle\langle 0^a|) U^{\dagger}$ where $a$ is the number of ancillas used.  In particular, there is only one Kraus operator $K=U$ in \Cref{eq:kraus_1}. In this case the number of $\Tgate$ gates in the circuit for $U$ is equal to $\mathcal{T}_{\mathrm{max}}(\mathcal{A})$ which is also equal to $\mathcal{T}(\mathcal{A})$.  

\paragraph*{Unitary Synthesis.}
For an $n$-qubit unitary $U$, consider the channel $\Ucal$ defined by
\begin{align}
\Ucal(\rho)\equiv U \rho U^{\dagger}\otimes |0^a\rangle \langle 0^a|.
\end{align}
We say that $\Acal$ implements $U$ to within error $\eps$ in the diamond distance if
\begin{align}\label{eq:unitary_diam_1}
\|\Acal-\Ucal\|_\diamond\leq \eps,
\end{align}
where $\|\cdot\|_\diamond$ here is the diamond norm.
As a strengthening of the lower bound in \Cref{thm:diagonal_unitary_T-count}, we will prove the following \Cref{thm:diagonal_adaptive_lb}.

\begin{theorem}\label{thm:diagonal_adaptive_lb}
There is an $n$-qubit diagonal unitary $D$ such that any adaptive Clifford+$\Tgate$ circuit $\Acal$ that implements $D$ up to error $\eps$ in diamond distance satisfies 
$$
\Tcal(\Acal)=\Omega\pbra{\sqrt{2^n\log(1/\eps)}+\log(1/\eps)}.
$$
\end{theorem}

Note that in the case where $\mathcal{A}$ is a unitary Clifford+$\Tgate$ circuit with no measurements, the diamond distance is equivalent to the usual distance with respect to the operator norm, up to a global phase.

Finally, we establish the following $\Tgate$-count lower bound for synthesis of general unitaries.

\begin{theorem}\label{thm:unitary_adaptive_lb}
There is an $n$-qubit unitary $U$ such that any adaptive Clifford+$\Tgate$ circuit $\Acal$ that implements $U$ up to error $\eps$ in diamond distance satisfies  
$$
\Tcal(\Acal)=\Omega\pbra{2^n\sqrt{\log(1/\eps)}+\log(1/\eps)}.
$$
\end{theorem}

The rest of the section is devoted to the proofs of \Cref{thm:state_adaptive_lb}, \Cref{thm:diagonal_adaptive_lb}, and \Cref{thm:unitary_adaptive_lb}.

\subsection{State preparation }\label{sec:state_prep_lb}

In this section we prove \Cref{thm:state_adaptive_lb}.

As an intermediate step in the proof, it will be convenient to work with a model of Clifford circuits with Pauli postselection, defined as follows.

An operator $P$ is an $n$-qubit Pauli iff $P=i^b P'$ for some $b\in\cbra{0,1,2,3}$ and $P'\in \{I,X,Y,Z\}^{\otimes n}$. In addition, if $b\in\cbra{0,2}$, then $P$ is a Hermitian Pauli. An $n$-qubit Pauli postselection is specified by an $n$-qubit Hermitian Pauli operator $P$.
On input state $\ket\phi$, the postselection produces a normalized state $\ket\psi$ of $\ket\phi$ projected to the $+1$ eigenspace of $P$, i.e., $\ket\psi$ is proportional to $(I+P)\ket\phi$, denoted by $\ket\psi\propto(I+P)\ket\phi$.
More precisely, $\ket\psi=\frac{(I+P)\ket\phi}{\vabs{(I+P)\ket\phi}}$ and it is defined arbitrarily if $(I+P)\ket\phi$ becomes zero.
Note that Pauli postselection is not a linear map.

\begin{definition}[Clifford circuit with Pauli postselection]\label{def:circuit_post}
A Clifford circuit with Pauli postselection is a sequence of Clifford gates and Pauli postselections applied to the input state.

A Clifford circuit with Pauli postselection and copies of the single-qubit magic state $|T\rangle=\frac{1}{\sqrt{2}}(|0\rangle+e^{i\pi/4}|1\rangle)$ can perform $\Tgate$ gates by gate injection.\footnote{In particular, we can implement a $\Tgate$ gate using the fact that $(I\otimes \langle 0|) \mathrm{CNOT}|\psi\rangle |T\rangle =\frac{1}{\sqrt{2}}T|\psi\rangle$ for any single qubit $|\psi\rangle$.} In the following we shall consider state preparation circuits of the following form: $\Ccal$ is a Clifford circuit with $m$ Pauli postselections, $n$ input qubits, and $a$ ancillas such that
\begin{equation}\label{eq:def:circuit_post_11}
\ket{\phi_\text{out}}\ket{0^{t+a}}=\Ccal\pbra{\ket{\phi_\text{in}}\ket{\Tgate}^{\otimes t}\ket{0^a}},
\end{equation}
 where $t$ is the number of magic states.
Note that we can explicitly write out \Cref{eq:def:circuit_post_11} as proportional to
\begin{equation}\label{eq:def:circuit_post_1}
C_{m+1}M_{m}C_{m}\cdots M_1C_1\pbra{\ket{\phi_\text{in}}\ket{\Tgate}^{\otimes t}\ket{0^a}},
\end{equation}
where each $C_j$ is an $(n+t+a)$-qubit Clifford and $M_j=I+P_j$ for some $(n+t+a)$-qubit Hermitian Pauli; and $C_j,M_j$ depend only on $\Ccal$, not on $\ket{\phi_\text{in}},\ket{\phi_\text{out}}$.
\end{definition}

The following Claim shows that, if there is an adaptive Clifford+$\Tgate$ circuit $\mathcal{A}$ that approximately prepares a given state with expected $\Tgate$-count $\mathcal{T}^0(\mathcal{A})$, then there is a Clifford circuit with Pauli postselections and $2 \mathcal{T}^0(\mathcal{A})$ magic states that also approximately prepares the state. This allows us to shift our focus to the model of Clifford circuits with Pauli postselection.  The proof uses an averaging argument similar to the one used in \cite[Lemma 5.4]{beverland2020lower}.

\begin{claim}\label{clm:state_to_post-sel}
Assume $\Acal$ is an adaptive Clifford+$\Tgate$ circuit that prepares an $n$-qubit state $\ket\psi$ up to error $\eps$ in trace distance.
Let $t=2\cdot\Tcal^0(\Acal)$.
Then for some integers $m,a\ge0$, there exists a Clifford circuit $\Ccal$ with $m$ Pauli postselections and $a$ ancillas such that $\Ccal\pbra{\ket{0^n}\ket{\Tgate}^{\otimes t}\ket{0^a}}=\ket\phi\ket{0^{t+a}}$ and $\frac12\vabs{\ketbra\psi\psi-\ketbra\phi\phi}_1\le\sqrt{6\eps}$.
\end{claim}
\begin{proof}
Assume $\Acal$ uses $a$ ancillas. Define $\rho=\ketbra{0^n}{0^n}$. Recall from \Cref{eq:kraus_2} that
\begin{align}\label{eq:clm:state_to_post-sel_1}
\Acal(\rho)=\sum_rp_r(\rho)\cdot\sigma_r(\rho),
\end{align}
where $p_r(\rho)$ is a distribution.
In addition, $K_r$ uses $\Tcal_r$ $\Tgate$ gates and 
\begin{align}\label{eq:clm:state_to_post-sel_2}
\sigma_r(\rho)=\ketbra{\pi_r}{\pi_r},
\end{align}
where $\ket{\pi_r}\propto K_r\pbra{\ket{0^n}\ket{0^a}}$ is a pure state.
We will find some $r$ such that $\ket{\pi_r}$ is close to $\ket\psi\ket{0^{a}}$ and $K_r$ uses a small number of $\Tgate$ gates.

By assumption, $\Acal(\rho)$ is $\eps$-close to $\ketbra\psi\psi\otimes\ketbra{0^a}{0^a}$ in trace distance.
Hence the fidelity $F$ between $\Acal(\rho)$ and $\ketbra\psi\psi\otimes\ketbra{0^a}{0^a}$ is
\begin{align}
F=\bra{0^a}\bra{\psi}\Acal(\rho)\ket\psi\ket{0^a}\ge(1-\eps)^2\ge1-2\eps,
\end{align}
where we use the inequality $1-\sqrt F\le\eps$ (see e.g., \cite[Theorem 9.3.1]{wilde2011classical}).
Plugging this into \Cref{eq:clm:state_to_post-sel_1} and \Cref{eq:clm:state_to_post-sel_2}, we have
\begin{equation}
\E_{r\sim p_r(\rho)}\Biggl[{\underbrace{\abs{\braket{\pi_r}\psi\ket{0^a}}^2}_{F_r}}\Biggr]\ge1-2\eps.
\end{equation}
Hence by Markov's inequality, for at least $2/3$ mass of $r$ (under distribution $p_r(\rho)$), the fidelity between $\ket{\pi_r}$ and $\ket\psi\ket{0^a}$ is $F_r\ge1-6\eps$.
On the other hand, since $\Tcal^0(\Acal)=\E_{r\sim p_r(\rho)}\sbra{\Tcal_r}$, for at least $1/2$ mass of $r$ we have $\Tcal_r\le2\cdot\Tcal^0(\Acal)$.
Hence there exists some $r$ such that $\Tcal_r\le2\cdot\Tcal^0(\Acal)$ and $F_r\ge1-6\eps$.

Fix any such $r$ and write $\ket{\pi_r}=\alpha\ket\phi\ket{0^a}+\sqrt{1-\alpha}\ket\bot$, where $\ket\phi$ is a normalized state and $(I\otimes\ketbra{0^a}{0^a})\ket\bot=0$.
Since 
\begin{equation}
\begin{aligned}
1-6\eps
&\le F_r=\abs{\braket{\pi_r}\psi\ket{0^a}}^2
=\alpha^2\abs{\braket\phi\psi}^2\le\abs{\braket\phi\psi}^2,
\end{aligned}
\end{equation}
we know that the fidelity between $\ket\phi$ and $\ket\psi$ is at least $1-6\eps$.
This implies that the trace distance between $\ket\phi$ and $\ket\psi$ is at most $\sqrt{6\eps}$ (see e.g., \cite[Theorem 9.3.1]{wilde2011classical}).

Finally we observe that $\ket{\phi}\ket{0^a}$ is obtained from $\ket{\pi_r}$ followed by postselecting the last $a$ qubits being all-zero.
This means we can construct $\ket\phi\ket{0^a}$ by applying $K_r$ and the above postselection on $\ket{0^n}\ket{0^a}$.
This process uses $t=2\cdot\Tcal^0(\Acal)$ $\Tgate$ gates,\footnote{Though $\Tcal_r\le2\cdot\Tcal(\Acal)$,  we can always ensure that there are exactly $t$ $\Tgate$ gates (say by applying dummy $\Tgate$ gates acting on one of the ancillas in the $|0\rangle$ state at the beginning of the circuit; note $T|0\rangle=|0\rangle$).} each of which can be implemented using a magic state $|T\rangle$ and gate injection. 
This gives a Clifford circuit $\Ccal$ with  Pauli postselection and input state $\ket{0^n}\ket{\Tgate}^{\otimes t}\ket{0^a}$ which prepares $|\phi\rangle|0^a\rangle$.
\end{proof}

Given \Cref{clm:state_to_post-sel}, we now focus on Clifford+$\Tgate$ circuits with Pauli postselections and ancillas. To complete the proof of \Cref{thm:state_adaptive_lb}, it is sufficient to prove the following proposition.

\begin{proposition}\label{prop:state_post_lb}
There is an $n$-qubit state $|\psi\rangle$ such that the following holds.
Assume $\Ccal$ is a Clifford circuit with $m$ Pauli postselections and $a$ ancillas.
Assume $\Ccal\pbra{\ket{0^n}{\ket T}^{\otimes t}\ket{0^a}}=\ket{\phi}\ket{0^{t+a}}$ and $\frac12\vabs{\ketbra\phi\phi-\ketbra\psi\psi}_1\le\eps$.
Then $t=\Omega\pbra{\sqrt{2^n\log(1/\eps)}+\log(1/\eps)}$.
\end{proposition}

The first lower bound of $\Omega(\log(1/\eps))$ holds even for single-qubit state preparation. This is proved by Beverland, Campbell, Howard, and Kliuchnikov \cite[Lemma 5.9]{beverland2020lower}.
The second $\Tgate$-count lower bound is a counting argument as in \cite{low2018trading}: the number of distinct circuits should be at least the number of distinguishable states in question.
The latter quantity follows directly from standard sphere packing bounds for $\ell_2$ balls (see e.g., \cite{wiki:Covering_number}).

\begin{fact}\label{fct:sphere_packing}
There are $N=(1/\eps)^{\Theta(2^n)}$ $n$-qubit states $\ket{\tau_1},\ldots,\ket{\tau_N}$ with pairwise trace distance at least $\eps$.
\end{fact}
\begin{proof}
We start with a set of $n$-qubit states $\ket{\zeta_1},\ldots,\ket{\zeta_M}$ that are pairwise $(10\eps)$-far in $\ell_2$ distance.
By the standard $\ell_2$ sphere packing, we can choose $M=(1/\eps)^{\Theta(2^n)}$.
Note that the trace distance between two pure states is asymptotically lower bounded by the $\ell_2$ distance with a global phase, i.e., 
\begin{align}
\frac12\vabs{\ketbra\phi\phi-\ketbra\psi\psi}_1
&=\sqrt{1-\abs{\braket\phi\psi}^2}
\ge\sqrt{1-\abs{\braket\phi\psi}}
=\frac1{\sqrt2}\sqrt{2-2\abs{\braket\phi\psi}}\\
&=\frac1{\sqrt2}\sqrt{2-2\max_{\theta\in[0,2\pi)}\Re\bigl(e^{i\theta}\braket\phi\psi\bigr)}\\
&=\frac1{\sqrt2}\cdot\min_{\theta\in[0,2\pi)}\vabs{\ket\phi-e^{i\theta}\ket\psi}.
\label{eq:fct:sphere_packing}
\end{align}
By discretizing the rotation angle $\theta$, each $\ket{\zeta_j}$ can be $\eps$-close to $O(1/\eps)$ other $\ket{\zeta_k}$'s in trace distance.
Hence we can find $N=\Theta(\eps M)=(1/\eps)^{\Theta(2^n)}$ states $\ket{\tau_1},\ldots,\ket{\tau_N}$ from $\ket{\zeta_1},\ldots,\ket{\zeta_M}$ such that their pairwise trace distance is at least $\eps$.
\end{proof}

To upper bound the number of distinct circuits, we express it in a canonical form as described in \Cref{lem:canonical_form_state}.
This was established in \cite[Section A.7]{beverland2020lower}.
Crucially, this canonical form has size independent of the number of Pauli postselections and ancillas.

\begin{lemma}[\cite{beverland2020lower}]\label{lem:canonical_form_state}
Let $\Ccal$ be a Clifford circuit with $m$ Pauli postselections, $n$ input qubits, and $a$ ancillas.
Let $\ket{\phi}$ be an $n$-qubit state.
Assume
\begin{equation}
\ket{\phi}\ket{0^{t+a}}=\Ccal\pbra{\ket{0^n}\ket{\Tgate}^{\otimes t}\ket{0^a}}.
\end{equation}

Then there exists an $(n+t)$-qubit Clifford unitary $C$ and matrices $M_1,M_2,\ldots,M_t$ such that
\begin{equation}
\label{eq:canonform}
\ket{\phi}\ket{0^t}\propto CM_t\cdots M_2M_1\pbra{\ket{0^n}\ket{\Tgate}^{\otimes t}},
\end{equation}
where each $M_j$ is $I+P_j$ for some $(n+t)$-qubit Hermitian Pauli $P_j$.
\end{lemma}

We remark that \Cref{lem:canonical_form_state} provides a canonical form for state preparation circuits. With a slight tweak, the original proof works for a more general setting including unitary synthesis circuits. For completeness and future reference, we present the statement as \Cref{thm:canonical_form} and prove it in \Cref{sec:canonical_form}.

We are now ready to prove \Cref{prop:state_post_lb}.
This, combined with \Cref{clm:state_to_post-sel}, completes the proof of \Cref{thm:state_adaptive_lb}.

\begin{proof}[Proof of \Cref{prop:state_post_lb}]
Assume $\eps<1/8$. This ensures that the trace distance between $\ket\phi$ and $\ket\psi$ is at most $1/8$ (see e.g., \cite[Section 9]{wilde2011classical}). Then the $\Omega(\log(1/\eps))$ lower bound follows from \cite[Lemma 5.9]{beverland2020lower}.

To get the $\Omega\pbra{\sqrt{2^n\log(1/\eps)}}$ lower bound,
assume for every $n$-qubit state $|\psi\rangle $ there is a Clifford circuit $\mathcal{C}$ with Pauli postselection such that $\Ccal\pbra{\ket{0^n}{\ket T}^{\otimes t}\ket{0^a}}=\ket{\phi}\ket{0^{t+a}}$ and the trace distance of $\ket\phi,\ket\psi$ is at most $\eps$.

By setting $\ket{\psi}=\ket{\tau_j}$ for each $j\in[N]$ from \Cref{fct:sphere_packing}, we apply \Cref{lem:canonical_form_state} to get a state preparation circuit in canonical form. Now let us count the number of possible state preparation circuits of the form \Cref{eq:canonform}. To this end note that there are $2^{O((n+t)^2)}$ Clifford unitaries $C$ on $n+t$ qubits,  and there are $2\cdot4^{n+t}$ choices for each Hermitian Pauli $P_j$.
So we have $2^{O(n^2+t^2)}$ possible state preparation circuits of the form \Cref{eq:canonform}.
By the construction of $\ket{\tau_j}$'s in \Cref{fct:sphere_packing}, each $\ket{\tau_j}$ gives a different circuit.
Thus $2^{O(n^2+t^2)}\ge(1/\eps)^{\Theta(2^n)}$, which gives $t=\Omega\pbra{\sqrt{2^n\log(1/\eps)}}$ since $\eps<1/8$.
\end{proof}

\subsection{Unitary synthesis}\label{sec:unitary_syn_lb}

Now we turn to the unitary synthesis $\Tgate$-count lower bounds (\Cref{thm:diagonal_adaptive_lb} and \Cref{thm:unitary_adaptive_lb}).
We will reduce them to state preparation tasks, by considering the action of an adaptive Clifford+$\Tgate$ circuit on the maximally entangled state and comparing this with the Choi state corresponding to the unitary $U$ of interest.

\begin{definition}[Maximally entangled state and Choi state]\label{def:max_choi}
We use 
$$
\ket\iota=2^{-n/2}\sum_{x\in\bin^n}\ket x\ket x
$$
to denote the $2n$-qubit maximally entangled state.
For an $n$-qubit unitary $U$, we use $\ket{\iota_U}$ to denote the Choi state, defined by
\begin{equation}
\ket{\iota_U}=2^{-n/2}\sum_{x\in\bin^n}\ket x\otimes\pbra{U\ket x}.
\end{equation}
\end{definition}

The following fact follows directly from \Cref{eq:unitary_diam_1} and the definition of the diamond norm.

\begin{fact}\label{fct:convert_to_state}
Assume an adaptive Clifford+$\Tgate$ circuit $\Acal$, using $a$ ancillas, implements an $n$-qubit unitary $U$ up to error $\eps$ in diamond distance.
Then $(I\otimes\Acal)(\ketbra\iota\iota)$ is $\eps$-close in trace distance to $\ketbra{\iota_U}{\iota_U}\otimes\ketbra{0^a}{0^a}$.
\end{fact}

By an almost identical argument as \Cref{clm:state_to_post-sel}, we have the following claim.

\begin{claim}\label{clm:unitary_to_post-sel}
Assume an adaptive Clifford+$\Tgate$ circuit $\Acal$ implements an $n$-qubit unitary $U$ up to error $\eps$ in the diamond distance.
Let $t=2\cdot\Tcal(\Acal)$.
Then for some integers $m,a\ge0$, there exists a Clifford circuit $\Ccal$ with $m$ Pauli postselections and $a$ ancillas such that $\Ccal\pbra{\ket{0^{2n}}\ket{\Tgate}^{\otimes t}\ket{0^a}}=\ket\phi\ket{0^{t+a}}$ and $\frac12\vabs{\ketbra{\iota_U}{\iota_U}-\ketbra\phi\phi}_1\le\sqrt{6\eps}$.
\end{claim}
\begin{proof}
By \Cref{fct:convert_to_state}, we know that $(I\otimes\Acal)(\ketbra\iota\iota)$ is $\eps$-close in trace distance to $\ketbra{\iota_U}{\iota_U}\otimes\ketbra{0^a}{0^a}$ for some integer $a\ge0$.
Then following the proof of \Cref{clm:state_to_post-sel} (but replacing the all-zeros input state with $\rho=\ketbra\iota\iota$ and  the expected $\Tgate$-count $\mathcal{T}^{0}(\mathcal{A})$ by the worst-case expected $\Tgate$-count $\mathcal{T}(\mathcal{A})$), we obtain a Clifford circuit $\Ccal'$ with Pauli postselection such that
\begin{equation}
\Ccal'\pbra{\ket\iota\ket{\Tgate}^{\otimes t}\ket{0^a}}=\ket\phi\ket{0^{t+a}},
\end{equation}
where the trace distance between $\ket{\iota_U}$ and $\ket\phi$ is at most $\sqrt{6\eps}$.
Moreover, $\Ccal'$ uses $t=2\cdot\Tcal(I\otimes\Acal)=2\cdot\Tcal(\Acal)$ magic states, where the second equality follows from \Cref{eq:t_complexity_1}.

Now notice that the maximally entangled state 
\begin{equation}
\ket\iota=2^{-n/2}\sum_{x\in\bin^n}\ket x\ket x=\pbra{\frac{\ket{00}+\ket{11}}{\sqrt2}}^{\otimes n}    
\end{equation}
can be constructed by $n$ single-qubit Hadamard gates and $n$ CNOTs on $\ket{0^{2n}}$.
Let $\Ccal$ be the circuit applying these Clifford operations followed by $\Ccal'$. Then we have $\Ccal\pbra{\ket{0^{2n}}\ket{\Tgate}^{\otimes t}\ket{0^a}}=\ket\phi\ket{0^{t+a}}$ as claimed.
\end{proof}

By \Cref{clm:unitary_to_post-sel}, it suffices to prove the following \Cref{prop:diag_post_lb} and \Cref{prop:unitary_post_lb} to obtain \Cref{thm:diagonal_adaptive_lb} and \Cref{thm:unitary_adaptive_lb} respectively.

\begin{proposition}\label{prop:diag_post_lb}
There is an $n$-qubit diagonal unitary $D$ such that the following holds.
Assume $\Ccal$ is a Clifford circuit with $m$ Pauli postselections and $a$ ancillas.
Assume 
$$
\Ccal\pbra{\ket{0^{2n}}{\ket T}^{\otimes t}\ket{0^a}}=\ket{\phi}\ket{0^{t+a}}
$$ 
and $\frac12\vabs{\ketbra\phi\phi-\ketbra{\iota_D}{\iota_D}}_1\le\eps$.
Then $t=\Omega\pbra{\sqrt{2^n\log(1/\eps)}+\log(1/\eps)}$.
\end{proposition}

\begin{proposition}\label{prop:unitary_post_lb}
There is an $n$-qubit unitary $U$ such that the following holds.
Assume $\Ccal$ is a Clifford circuit with $m$ Pauli postselections and $a$ ancillas.
Assume 
$$
\Ccal\pbra{\ket{0^{2n}}{\ket T}^{\otimes t}\ket{0^a}}=\ket{\phi}\ket{0^{t+a}}
$$
and $\frac12\vabs{\ketbra\phi\phi-\ketbra{\iota_U}{\iota_U}}_1\le\eps$.
Then $t=\Omega\pbra{2^n\sqrt{\log(1/\eps)}+\log(1/\eps)}$.
\end{proposition}

Following the proof of \Cref{thm:state_adaptive_lb}, we now need to construct a large set of distinguishable Choi states to apply the counting argument.
For this purpose we relate the trace distance between Choi states to normalized Hilbert–Schmidt (aka Frobenius) distance of the corresponding unitaries.

In particular, if $A$ is a $d\times d$ matrix we write $\hsnorm{A}=\sqrt{\mathrm{Tr}(A^{\dagger} A)/d}$.
\begin{fact}\label{fct:choi_to_unitary}
Let $U$ and $V$ be two $n$-qubit unitaries.
The trace distance between their Choi states is
\begin{align}
\frac12\vabs{\ketbra{\iota_U}{\iota_U}-\ketbra{\iota_V}{\iota_V}}_1
&\ge\frac1{\sqrt2}\cdot\min_{\theta\in[0,2\pi)}\hsnorm{U-e^{i\theta}\cdot V},
\end{align}
where $\hsnorm{\cdot}$ is the normalized Hilbert-Schmidt norm.
\end{fact}
\begin{proof}
By \Cref{eq:fct:sphere_packing}, it suffices to observe the following
\begin{align}
\vabs{\ket{\iota_U}-e^{i\theta}\ket{\iota_V}}^2
&=\biggl\|{2^{-n/2}\sum_{x\in\bin^n}\ket x\otimes\pbra{\pbra{U-e^{i\theta}V}\ket x}}\biggr\|^2\\
&=2^{-n}\sum_{x\in\bin^n}\vabs{\pbra{U-e^{i\theta}V}\ket x}^2\\
&=2^{-n}\cdot\tr\pbra{\pbra{U-e^{i\theta}V}^\dag{U-e^{i\theta}V}}\\
&=\hsnorm{U-e^{i\theta}V}^2.
\tag*{\qedhere}
\end{align}
\end{proof}

Note that we have similar packing bounds.

\begin{fact}\label{fct:diag_sphere_packing}
There are $N=(1/\eps)^{\Theta(2^n)}$ $n$-qubit diagonal unitaries $D_1,\ldots,D_N$ such that the pairwise trace distance between their Choi states is at least $\eps$.
\end{fact}
\begin{proof}
By \Cref{fct:choi_to_unitary}, we work with the normalized Hilbert-Schmidt norm.
The phase issue will be handled by discretization, same as the proof of \Cref{fct:sphere_packing}.

Observe that the normalized Hilbert-Schmidt distance between two diagonal unitaries is the normalized $\ell_2$ distance between their diagonals (viewed as vectors of length $2^n$).
Then by standard sphere packing results, we can get $M=(1/\eps)^{\Theta(2^n)}$ diagonal unitaries $\tilde D_1,\ldots,\tilde D_M$ that are pairwise $(10\eps)$-far in normalized Hilbert-Schmidt distance.
By discretizing the rotation angle of the global phase, each $\tilde D_j$ can be $(\sqrt2\eps)$-close to $O(1/\eps)$ other $\tilde D_k$'s in normalized Hilbert-Schmidt distance even with the presence of a global phase.
This, combined with \Cref{fct:choi_to_unitary}, implies that we can find $N=\Theta(\eps M)=(1/\eps)^{\Theta(2^n)}$ diagonal unitaries $D_1,\ldots,D_N$ from $\tilde D_1,\ldots,\tilde D_M$ such that the pairwise trace distance between their Choi states is at least $\eps$.
\end{proof}

\begin{fact}\label{fct:unitary_sphere_packing}
There are $N=(1/\eps)^{\Theta(4^n)}$ $n$-qubit unitaries $U_1,\ldots,U_N$ such that the pairwise trace distance between their Choi states is at least $\eps$.
\end{fact}
\begin{proof}
We only show how to obtain $M=(1/\eps)^{\Theta(4^n)}$ unitaries that are $(\sqrt2\eps)$-far from each other in normalized Hilbert-Schmidt distance.
The remaining proof is the same as \Cref{fct:diag_sphere_packing}.

Let $A$ be a skew Hermitian matrix (i.e., $A+A^\dag=0$). Then $e^A$ is a unitary.
By relating the normalized Hilbert-Schmidt distance before and after matrix exponentiation, we construct the desired unitaries via ``well-separated skew Hermitian matrices'' (see e.g., \cite[Appendix D]{barthel2018fundamental}).
To this end, let $B$ be another skew Hermitian matrix and assume $\hsnorm{A},\hsnorm{B}\le1/2$.
\begin{samepage}
Then notice that
\begin{align}
\hsnorm{e^A-e^B}
&=\hsnorm{\sum_{k=1}^{+\infty}\frac{A^k-B^k}{k!}}
\tag{Taylor expansion}\\
&\ge\hsnorm{A-B}-\sum_{k\ge2}\frac{\hsnorm{A^k-B^k}}{k!}\\
&=\hsnorm{A-B}-\sum_{k\ge2}\frac{\hsnorm{\sum_{\ell=1}^kA^{k-\ell}(A-B)B^{\ell-1}}}{k!}\\
&\ge\hsnorm{A-B}-\sum_{k\ge2}\frac{k\cdot2^{-(k-1)}\cdot\hsnorm{A-B}}{k!}
\tag{since $\hsnorm{A},\hsnorm{B}\le1/2$}\\
&=(2-\sqrt e)\hsnorm{A-B}
\ge\hsnorm{A-B}/3.
\end{align}
\end{samepage}
Hence we just need to obtain skew Hermitian matrices $A_1,\ldots,A_M$ of norm $\hsnorm{A_j}\le1/2$ and pairwise distance $\hsnorm{A_j-A_k}\ge3\sqrt2\eps$.
The desired bound on $M$ follows from standard sphere packing in $\ell_2$ norm, by viewing the upper triangular part of a $2^n\times2^n$ skew Hermitian matrices as a vector of length $2^n(2^n+1)/2=\Theta(4^n)$.
\end{proof}

At this point, we use the canonical circuit from \Cref{lem:canonical_form_state} and prove \Cref{prop:diag_post_lb} and \Cref{prop:unitary_post_lb}, which, combined with \Cref{clm:unitary_to_post-sel}, complete the proof of \Cref{thm:diagonal_adaptive_lb} and \Cref{thm:unitary_adaptive_lb}.

\begin{proof}[Proof of \Cref{prop:diag_post_lb}]
Assume $\eps$ is a small constant.
Then the $\Omega(\log(1/\eps))$ lower bound also follows from \cite[Lemma 5.9]{beverland2020lower}.
To get the $\Omega\pbra{\sqrt{2^n\log(1/\eps)}}$ lower bound, we follow the same calculation as in the proof of \Cref{prop:state_post_lb}.
Suppose that for every diagonal unitary $D$ there is a Clifford circuit $\mathcal{C}$ with Pauli postselection such that  $\Ccal\pbra{\ket{0^{2n}}{\ket T}^{\otimes t}\ket{0^a}}=\ket{\phi}\ket{0^{t+a}}$ and the trace distance between $\ket\phi$ and $\ket{\iota_D}$ is at most $\eps$.

In light of \Cref{fct:choi_to_unitary}, we consider the diagonal unitaries $D=D_j$ for each $j\in[N]$ from \Cref{fct:diag_sphere_packing}.
Then by \Cref{lem:canonical_form_state}, we obtain a canonical form  for the $(2n+t)$-qubit Clifford circuit with Pauli postselection that approximately prepares $\ket{\iota_{D_j}}$ for each $j$. But there are at most $2^{O(n^2+t^2)}$ different canonical forms on $2n+t$ qubits.
By \Cref{fct:choi_to_unitary} and \Cref{fct:diag_sphere_packing}, each $\ket{\iota_{D_j}}$ gives a different circuit.
Thus $2^{O(n^2+t^2)}\ge(1/\eps)^{\Theta(2^n)}$, which gives $t=\Omega\pbra{\sqrt{2^n\log(1/\eps)}}$ since $\eps<1/8$.
\end{proof}

\begin{proof}[Proof of \Cref{prop:unitary_post_lb}]
As above, the $\Omega(\log(1/\eps))$ lower bound follows from \cite{beverland2020lower}.
The $\Omega\pbra{2^n\sqrt{\log(1/\eps)}}$ follows by applying the same counting argument used in the proof of \Cref{prop:diag_post_lb}, but now with the unitaries $U=U_j$ for each $j\in[N]$ from \Cref{fct:unitary_sphere_packing} and $N=(1/\eps)^{\Theta(4^n)}$.
\end{proof}

\section*{Acknowledgments}
We thank Ryan Babbush, Dominic Berry, Michael Beverland, Oscar Higgott, Tanuj Khattar, Guang Hao Low, Dmitri Maslov, Luke Shaeffer, and Rolando Somma for helpful discussions. 
We thank Tanuj Khattar for pointing out that the analysis of the algorithm in \cite{low2018trading} can be improved.
We thank anonymous reviewers for helpful comments.
KW wants to thank Daniel Grier, Jiaqing Jiang, Saeed Mehraban, and Anurudh Peduri for relevant references.

DG is a CIFAR fellow in the quantum information science program.
This work was done while KW was at Google.

\section*{Author contribution}
The authors jointly discussed the technical contents of the paper and wrote the paper, including all the revisions.
Large language model tools were only used to address typesetting and package collision issues when switching to this template in the final stage, after which the authors also proofread.

\bibliographystyle{alphaurl} 
\bibliography{ref}

\appendix

\section{Canonical form of Clifford circuits with magic states and Pauli postselection}\label{sec:canonical_form}

To describe our generalization of \Cref{lem:canonical_form_state}, we need to recall the definition of stabilizer group.

\begin{definition}[Stabilizer group]\label{def:stabilizer}
For a state $\ket\tau$, its stabilizer group $\Stab(\ket\tau)$ is the set of Pauli operators $P$ satisfying $P\ket\tau=\ket\tau$.
For a set $\Lambda$ of states, we define its stabilizer group as $\Stab(\Lambda)=\bigcap_{\ket\tau\in\Lambda}\Stab(\ket\tau)$.
Note that the stabilizer group is a set of commutative Hermitian Pauli operators and does not contain $-I$.
We also remark that the size of a stabilizer group is an integer power of two.
\end{definition}

Now we present the formal statement for the canonical form as \Cref{thm:canonical_form}.

\begin{theorem}\label{thm:canonical_form}
Let $\Ccal$ be a Clifford circuit with $m$ Pauli postselections, $n$ input qubits, and $a$ ancillas.
Let $\cbra{\ket{\phi_\lambda}}_{\lambda\in\Scal}$ and $\cbra{\ket{\psi_\lambda}}_{\lambda\in\Scal}$ be two sets of $n$-qubit states indexed by $\Scal$.
Assume
\begin{equation}\label{eq:thm:canonical_form_1}
\ket{\psi_\lambda}\ket{0^{t+a}}=\Ccal\pbra{\ket{\phi_\lambda}\ket{\Tgate}^{\otimes t}\ket{0^a}}
\end{equation}
holds for all $\lambda\in\Scal$,
where we recall the circuit expression in \Cref{def:circuit_post}.
Then there exists an $(n+t)$-qubit Clifford unitary $C$ and matrices $M_1,M_2,\ldots,M_c$ such that
\begin{equation}\label{eq:thm:canonical_form_2}
\ket{\psi_\lambda}\ket{0^t}\propto CM_c\cdots M_2M_1\pbra{\ket{\phi_\lambda}\ket{\Tgate}^{\otimes t}}
\end{equation}
holds for all $\lambda\in\Scal$,
where each $M_j$ is $I+P_j$ for some $(n+t)$-qubit Hermitian Pauli $P_j$ and
\begin{align}
c&=t-\log\pbra{\abs{\Stab\pbra{\cbra{\ket{\phi_\lambda}}_{\lambda\in\Scal}}}}
+\log\pbra{\abs{\Stab\pbra{\cbra{\ket{\psi_\lambda}}_{\lambda\in\Scal}}}}.
\end{align}
\end{theorem}

Intuitively, \Cref{thm:canonical_form} says that the number of ancillary qubits and the number of Pauli postselections are upper bounded by the number of $\Tgate$ gates.
\begin{itemize}
\item If $\Scal$ is a singleton set and $\ket\phi=\ket{0^n}$, then $\Stab(\ket\phi)=2^n$ and $\Stab(\ket\psi)\le2^n$.
Thus $c\le t$ and \Cref{thm:canonical_form} specializes as \Cref{lem:canonical_form_state}, which justifies the above intuition for state preparation task, as also proved in \cite{beverland2020lower}.
\item If $\cbra{\ket{\phi_\lambda}}_{\lambda\in\Scal}$ ranges over all $n$-qubit states with $\ket{\psi_\lambda}=V\ket{\phi_\lambda}$ for an $n$-qubit unitary $V$, then $\abs{\Stab\pbra{\cbra{\ket{\phi_\lambda}}_{\lambda\in\Scal}}}=\abs{\Stab\pbra{\cbra{\ket{\psi_\lambda}}_{\lambda\in\Scal}}}=1$.
Thus $c=t$ and \Cref{thm:canonical_form} justifies the above intuition for unitary synthesis of $V$.
\end{itemize}
For intermediate tasks like implementing the first $K$ columns of a unitary (aka partial unitary synthesis), \Cref{thm:canonical_form} can also be used, by setting $\cbra{\ket{\phi_\lambda}}_{\lambda\in\Scal}$ to be the first $K$ computational basis states and analyzing the stabilizer group of the corresponding outputs $\cbra{\ket{\psi_\lambda}}_{\lambda\in\Scal}$.

The proof of \Cref{thm:canonical_form} follows the analysis in \cite[Section A.7]{beverland2020lower}.
We start with the following convenient fact.

\begin{fact}\label{fct:anticom_pauli_cliff}
Let $P,Q$ be two anticommuting Hermitian Pauli operators.
Then $-iPQ$ is Hermitian Pauli and $\frac{I+PQ}{\sqrt2}$ is Clifford.
\end{fact}
\begin{proof}
Consider the Clifford that conjugates $P$ to $\Xgate\otimes I$ and $Q$ to $\Zgate\otimes I$. 
Then the result follows.
\end{proof}

In light of \Cref{eq:def:circuit_post_1} and \Cref{eq:thm:canonical_form_2}, we prove the following lemma.
\begin{lemma}\label{lem:organize}
Let $C_1,\ldots,C_{r+1}$ be $K$-qubit Clifford unitaries and $P_1,\ldots,P_r$ be $K$-qubit Hermitian Pauli operators.
Define $M_j=I+P_j$ for each $j\in[r]$.
Let $\cbra{\ket{\rho_\lambda}}_{\lambda\in\Ecal}$ and $\cbra{\ket{\tau_\lambda}}_{\lambda\in\Ecal}$ be two sets of $K$-qubit states indexed by $\Ecal$.
Assume 
\begin{equation}\label{eq:lem:organize_1}
\ket{\tau_\lambda}\propto C_{r+1}M_rC_r\cdots M_1C_1\ket{\rho_\lambda}
\end{equation}
holds for all $\lambda\in\Ecal$.
Then there exist an integer $0\le r'\le r$, a $K$-qubit Clifford $C$, and $K$-qubit Hermitian Pauli operators $Q_1,\ldots,Q_{r'}$ such that the following holds.
\begin{enumerate}
\item\label{itm:lem:organize_1} 
Let $A_j=I+Q_j$.
Then $\ket{\tau_\lambda}\propto CA_{r'}\cdots A_2A_1\ket{\rho_\lambda}$ holds for all $\lambda\in\Ecal$.
\item\label{itm:lem:organize_2}
Let $\Gcal=\Stab\pbra{\cbra{\ket{\rho_\lambda}}_{\lambda\in\Ecal}}$.
For each $j\in[r']$, we have $Q_j\notin\abra{Q_1,\ldots,Q_{j-1},\Gcal}$ and $Q_j$ commutes with any $P\in\abra{Q_1,\ldots,Q_{j-1},\Gcal}$.
\item\label{itm:lem:organize_3}
$\abs{\Stab\pbra{\cbra{\ket{\rho_\lambda}}_{\lambda\in\Ecal}}}\cdot2^{r'}\le\abs{\Stab\pbra{\cbra{\ket{\tau_\lambda}}_{\lambda\in\Ecal}}}$.
\end{enumerate}
\end{lemma}
\begin{proof}
Observe that
\begin{align}
C_{r+1}M_rC_r\cdots M_1C_1
&=
C_{r+1}M_rC_r\cdots M_2C_2C_1\pbra{C_1^\dag M_1C_1}
\label{eq:lem:organize_2}\\
&=C_{r+1}M_rC_r\cdots M_2C_2C_1A_1
\tag{define $A_1=C_1^\dag M_1C_1$}\\
&=\cdots=\\
&=C_{r+1}\cdots C_2C_1\cdot A_r\cdots A_2A_1
\tag{define $A_j=(C_j\cdots C_1)^\dag M_j(C_j\cdots C_1)$}\\
&=CA_r\cdots A_2A_1.
\tag{define $C=C_{t+1}\cdots C_2C_1$}
\end{align}
Since each $C_j$ is Clifford, the final $C$ is also Clifford.
Since each $M_j=I+P_j$ for some Hermitian Pauli $P_j$, we have $A_j=I+(C_j\cdots C_1)^\dag P_j(C_j\cdots C_1)=:I+Q_j$ where $Q_j$ is also a Hermitian Pauli.
Combined with \Cref{eq:lem:organize_1}, this already guarantees \Cref{itm:lem:organize_1}.

To ensure \Cref{itm:lem:organize_2}, we perform a clean-up procedure inductively for $j=1,2,\ldots$.
To this end, we divide into the following cases.
\begin{itemize}
\item If $Q_j$ anticommutes with some $R\in\Gcal$, then we show it can be discarded.
Observe that 
\begin{align}
(Q_jR)A_{j-1}\cdots A_1\ket{\rho_\lambda}
&=Q_jA_{j-1}\cdots A_1R\ket{\rho_\lambda}
\tag{by induction and \Cref{itm:lem:organize_2}}\\
&=Q_jA_{j-1}\cdots A_1\ket{\rho_\lambda}.
\tag{since $R\in\Gcal\subseteq\Stab(\ket{\rho_\lambda})$}
\end{align}
Hence we can safely replace $A_j=I+Q_j$ as $I+Q_jR\propto\frac{I+Q_jR}{\sqrt2}$, which is Clifford due to \Cref{fct:anticom_pauli_cliff}.
Then we can apply the same trick in \Cref{eq:lem:organize_2} to merge with the Clifford $C$ outside.
\item If $Q_j$ anticommutes with $Q_{j^*}$ for some $j^*<j$, then it can also be discarded.
By induction and \Cref{itm:lem:organize_1}, we assume without loss of generality $j^*=j-1$.
Since $Q_{j-1}$ is Hermitian Pauli, we have $Q_{j-1}A_{j-1}=Q_{j-1}(I+Q_{j-1})=Q_{j-1}+I=A_{j-1}$.
Therefore we can safely replace $A_j=I+Q_j$ as $\frac{I+Q_jQ_{j-1}}{\sqrt2}$, which again is Clifford by \Cref{fct:anticom_pauli_cliff}.
\item Now we are in the situation where $Q_j$ commutes with $Q_1,\ldots,Q_{j-1}$ and elements in $\Gcal$.
To guarantee \Cref{itm:lem:organize_2}, it suffices to show that if $Q_j\in\abra{Q_1,\ldots,Q_{j-1},\Gcal}$, then it can be discarded.
Assume $Q_j\in\abra{Q_1,\ldots,Q_{j-1},\Gcal}$.
Since $Q_1,\ldots,Q_{j-1}$ and $\Gcal$ are commuting Hermitian Pauli operators, $Q_j$ can be expressed as $R\cdot\prod_{\ell\in L}Q_\ell$ for some $L\subseteq[j-1]$ and $R\in\Gcal$.
This means $Q_j$ lies in the stabilizer group of $A_{j-1}\cdots A_1\ket{\rho_\lambda}$ and hence
\begin{align}
Q_jA_{j-1}\cdots A_1\ket{\rho_\lambda}
=A_{j-1}\cdots A_1\ket{\rho_\lambda}.
\end{align}
This means $A_jA_{j-1}\cdots A_2A_1\ket{\rho_\lambda}\propto A_{j-1}\cdots A_2A_1\ket{\rho_\lambda}$ and thus it can be removed directly.
\end{itemize}

\begin{samepage}
Finally we prove \Cref{itm:lem:organize_3}.
By \Cref{itm:lem:organize_1}, we have
\begin{align}
\abs{\Stab\pbra{\cbra{\ket{\tau_\lambda}}_{\lambda\in\Ecal}}}
&=\abs{\Stab\pbra{\cbra{CA_{r'}\cdots A_1\ket{\rho_\lambda}}_{\lambda\in\Ecal}}}\\
&=\abs{\Stab\pbra{\cbra{A_{r'}\cdots A_1\ket{\rho_\lambda}}_{\lambda\in\Ecal}}}\\
&=:|\Hcal|.
\tag{since $C$ is Clifford}
\end{align}
\end{samepage}
Since $Q_jA_j=A_j$, we have $Q_j\in\Hcal$ for each $j\in[r']$.
We also have $P\in\Hcal$ for each $P\in\Gcal$.
Therefore, $\abra{Q_1,\ldots,Q_{r'},\Gcal}\subseteq\Hcal$, where, by \Cref{itm:lem:organize_2}, the former set has size $2^{r'}\cdot|\Gcal|$.
Hence $|\Hcal|\ge2^{r'}\cdot|\Gcal|$. Putting back the definition of $\Gcal$ and $\Hcal$ implies \Cref{itm:lem:organize_3}.
\end{proof}

In addition to \Cref{lem:organize}, we need the following technical lemma to remove unnecessary ancillas.
\Cref{lem:generalize_beverlend_lemA.10} generalizes \cite[Lemma A.10]{beverland2020lower} and shares a similar proof.

\begin{restatable}{lemma}{decouplestab}\label{lem:generalize_beverlend_lemA.10}
Let $\cbra{\ket{\phi_\lambda}}_{\lambda\in\Scal}$ and $\cbra{\ket{\psi_\lambda}}_{\lambda\in\Scal}$ be two sets of $n$-qubit states indexed by $\Scal$.
Let $C'$ be an $(n+m)$-qubit Clifford unitary where $m\ge0$ is an integer.
Assume
\begin{align}\label{eq:lem:generalize_beverlend_lemA.10_1}
\ket{\psi_\lambda}\ket{0^m}=C'\pbra{\ket{\phi_\lambda}\ket{0^m}}
\quad\text{holds for all $\lambda\in\Scal$.}
\end{align}
Then there exists an $n$-qubit Clifford unitary $C$ such that $\ket{\psi_\lambda}=C\ket{\phi_\lambda}$ for all $\lambda\in\Scal$.
\end{restatable}

Given \Cref{lem:generalize_beverlend_lemA.10} and \Cref{lem:organize}, we first complete the proof of \Cref{thm:canonical_form}. The proof of \Cref{lem:generalize_beverlend_lemA.10} is presented in \Cref{sec:decouple_stab}.

\begin{proof}[Proof of \Cref{thm:canonical_form}]
By \Cref{def:circuit_post} and \Cref{eq:thm:canonical_form_1}, $\Ccal$ can be expressed by $(n+t+a)$-qubit Clifford $C_1,\ldots,C_{m+1}$ and $(n+t+a)$-qubit Hermitian Pauli $R_1,\ldots,R_{m}$ such that
\begin{align}
\ket{\psi_\lambda}\ket{0^t}\ket{0^a}
&\propto
C_{m+1}W_{m}C_{m}\cdots W_1C_1\pbra{\ket{\phi_\lambda}\ket{\Tgate}^{\otimes t}\ket{0^a}}
\end{align}
holds for all $\lambda\in\Scal$,
where $W_j=I+R_j$.

By \Cref{lem:organize}, we can further simplify above by an integer $0\le c\le m$, an $(n+t+a)$-qubit Clifford $C'$, and $(n+t+a)$-qubit Hermitian Pauli $Q_1,\ldots,Q_c$ such that
\begin{align}\label{eq:thm:canonical_form_5}
\ket{\psi_\lambda}\ket{0^t}\ket{0^a}
\propto
C'A_c\cdots A_2A_1\pbra{\ket{\phi_\lambda}\ket{\Tgate}^{\otimes t}\ket{0^a}}
\end{align}
holds for all $\lambda\in\Scal$,
where $A_j=I+Q_j$.
In addition,
\begin{align}
2^{t+a-c}\cdot\abs{\Stab\pbra{\cbra{\ket{\psi_\lambda}}_{\lambda\in\Scal}}}
&=2^{-c}\cdot\abs{\Stab\pbra{\cbra{\ket{\psi_\lambda}\ket{0^t}\ket{0^a}}_{\lambda\in\Scal}}}
\tag{since $\Stab(\ket0)=\cbra{I,Z}$}\\
&\ge\abs{\Stab\pbra{\cbra{\ket{\phi_\lambda}\ket{\Tgate}^{\otimes t}\ket{0^a}}_{\lambda\in\Scal}}}
\tag{by \Cref{itm:lem:organize_3} of \Cref{lem:organize}}\\
&\ge\abs{\Stab\pbra{\cbra{\ket{\phi_\lambda}}_{\lambda\in\Scal}}}\cdot\abs{\Stab\pbra{\ket{0^a}}}\\
&=2^a\cdot\abs{\Stab\pbra{\cbra{\ket{\phi_\lambda}}_{\lambda\in\Scal}}},
\end{align}
which implies
\begin{align}\label{eq:thm:canonical_form_3}
c&\le t-\log\pbra{\abs{\Stab\pbra{\cbra{\ket{\phi_\lambda}}_{\lambda\in\Scal}}}}
+\log\pbra{\abs{\Stab\pbra{\cbra{\ket{\psi_\lambda}}_{\lambda\in\Scal}}}}.
\end{align}
By \Cref{itm:lem:organize_2} of \Cref{lem:organize}, we know that each $Q_j$ commutes with $\Stab\pbra{\cbra{\ket{\phi_\lambda}\ket{\Tgate}^{\otimes t}\ket{0^a}}_{\lambda\in\Scal}}$, which contains all $\Zgate$ rotations on the last $a$-qubits.
Since $Q_j$ is a Hermitian Pauli, we have
\begin{equation}
Q_j=P_j\otimes\Zgate(v_j)
\end{equation}
for some $(n+t)$-qubit Hermitian Pauli $P_j$ and some $v_j\in\bin^a$.
Here, $\Zgate(v_j)$ applies the single-qubit $\Zgate$ gate on the coordinates selected by $v_j$.
Define $M_j=I+P_j$.
Then for any $(n+t)$-qubit state $\ket\rho$, we have
\begin{align}\label{eq:thm:canonical_form_4}
A_j\pbra{\ket{\rho}\ket{0^a}}=\pbra{M_j\ket{\rho}}\otimes\ket{0^a}.
\end{align}
Therefore,
\begin{align}
\ket{\psi_\lambda}\ket{0^t}\ket{0^a}
&\propto
C'A_c\cdots A_2A_1\pbra{\ket{\phi_\lambda}\ket{\Tgate}^{\otimes t}\ket{0^a}}
\tag{by \Cref{eq:thm:canonical_form_5}}\\
&=C'\pbra{\pbra{M_c\cdots M_2M_1\ket{\phi_\lambda}\ket{\Tgate}^{\otimes t}}\otimes\ket{0^a}}.
\tag{by \Cref{eq:thm:canonical_form_4}}
\end{align}
Now define $\ket{\rho_\lambda}=M_c\cdots M_2M_1\ket{\phi_\lambda}\ket{\Tgate}^{\otimes t}$.
Then for all $\lambda\in\Scal$, 
$$
\ket{\psi_\lambda}\ket{0^t}\otimes\ket{0^a}\propto C'\pbra{\ket{\rho_\lambda}\otimes\ket{0^{a}}}.
$$
where $C'$ is Clifford.
Hence by \Cref{lem:generalize_beverlend_lemA.10}, there exists an $(n+t)$-qubit Clifford $C$ such that 
\begin{equation}
\ket{\psi_\lambda}\ket{0^t}\propto C\ket{\rho_\lambda}=CM_c\cdots M_2M_1\pbra{\ket{\phi_\lambda}\ket{\Tgate}^{\otimes t}}.
\end{equation}
To complete the proof of \Cref{thm:canonical_form}, we finally remark that it is always possible to pad dummy $P=I$ to ensure the equality in \Cref{eq:thm:canonical_form_3} holds.
\end{proof}

\subsection{Eliminating ancillas in Clifford circuits}\label{sec:decouple_stab}

Finally we prove \Cref{lem:generalize_beverlend_lemA.10}, which generalizes \cite[Lemma A.10]{beverland2020lower}.

\decouplestab*

\begin{proof}
Let $r=\log\abs{\Stab(\cbra{\ket{\phi_\lambda}})}$.
Since $\Stab(\cbra{\ket{\phi_\lambda}})$ is a commutative subgroup of the Pauli group and does not contain $-I$, there exists a Clifford $C_\phi$ conjugating $\Stab(\cbra{\ket{\phi_\lambda}})$ to the group of $\Zgate$ operators on the last $r$ qubits \cite{cleve1997efficient}, which means
\begin{align}\label{eq:lem:generalize_beverlend_lemA.10_1.1}
\ket{\phi_\lambda}=C_\phi\pbra{\ket{\phi_\lambda'}\ket{0^r}}
\end{align}
holds for every $\lambda\in\Scal$ and some $(n-r)$-qubit state $\ket{\phi_\lambda'}$.
By \Cref{eq:lem:generalize_beverlend_lemA.10_1}, we also have 
$$
\log\abs{\Stab(\cbra{\ket{\psi_\lambda}})}=r
$$
and analogously, there exists a Clifford $C_\psi$ such that
\begin{align}\label{eq:lem:generalize_beverlend_lemA.10_1.2}
\ket{\psi_\lambda}=C_\psi\pbra{\ket{\psi_\lambda'}\ket{0^r}}
\end{align}
holds for every $\lambda\in\Scal$ and some $(n-r)$-qubit state $\ket{\psi_\lambda'}$.
Now define Clifford $C''=(C_\psi\otimes I_m)^\dag C'(C_\phi\otimes I_m)$.
Then \Cref{eq:lem:generalize_beverlend_lemA.10_1} becomes
\begin{align}\label{eq:lem:generalize_beverlend_lemA.10_2}
\ket{\psi_\lambda'}\ket{0^{r+m}}=C''\pbra{\ket{\phi_\lambda'}\ket{0^{r+m}}}
\end{align}
holds for all $\lambda\in\Scal$.
In addition, $\Stab(\cbra{\ket{\phi_\lambda'}})$ and $\Stab(\cbra{\ket{\psi_\lambda'}})$ contain only the trivial stabilizer $I$.

We will make another simplification to make sure that $C''$ commutes with the single-qubit $\Zgate$ gate on the last qubit, which will help factorize $C''$ and apply induction.
To this end, we observe that $\Stab(\cbra{\ket{\phi_\lambda'}\ket{0^{r+m}}})=\Stab(\cbra{\ket{\psi_\lambda'}\ket{0^{r+m}}})$ equals the Pauli $\Zgate$ group on the last $r+m$ qubits.
Hence by \Cref{eq:lem:generalize_beverlend_lemA.10_2}, for the single-qubit $\Zgate$ gate on the last qubit we know that
\begin{equation}
C''(I_{n+m-1}\otimes \Zgate)C''^\dag=I_{n-r}\otimes \Zgate(s)
\end{equation}
for some $s\in\bin^{r+m}$.
Let $D$ be an $(r+m)$-qubit CNOT circuit reversing this mapping, i.e., $D\Zgate(s)D^\dag$ equals the single-qubit $\Zgate$ gate on the last qubit.
Define Clifford $C'''=(I_{n-r}\otimes D)C''$. 
Then we have the commutativity condition:
\begin{align}\label{eq:lem:generalize_beverlend_lemA.10_3}
C'''(I_{n+m-1}\otimes \Zgate)C'''^\dag
&=(I_{n-r}\otimes D)(I_{n-r}\otimes \Zgate(s))(I_{n-r}\otimes D)^\dag\\
&=I_{n+m-1}\otimes \Zgate.
\end{align}
Moreover, \Cref{eq:lem:generalize_beverlend_lemA.10_2} remains unchanged:
\begin{align}\label{eq:lem:generalize_beverlend_lemA.10_4}
C'''\pbra{\ket{\phi_\lambda'}\ket{0^{r+m}}}
&=(I_{n-r}\otimes D)\ket{\psi_\lambda'}\ket{0^{r+m}}
=\ket{\psi_\lambda'}\ket{0^{r+m}}
\end{align}
holds for all $\lambda\in\Scal$,
where the second equality uses the fact that $D$ is a CNOT circuit.

At this point, we can factorize $C'''$ to reduce the number of ancillary qubits.
This is achieved by the following \Cref{fct:beverland_propA.11} in \cite{beverland2020lower}.

\begin{fact}[{\cite[Proposition A.11 and Lemma A.13]{beverland2020lower}}]\label{fct:beverland_propA.11}
Assume $C$ is a $K$-qubit Clifford unitary commuting with the single-qubit $\Zgate$ gate on the last qubit.
Then $C=C_0\otimes\ketbra00+C_1\otimes\ketbra11$, where $C_0,C_1$ are $(K-1)$-qubit Clifford unitaries.
\end{fact}

Combining \Cref{fct:beverland_propA.11} and \Cref{eq:lem:generalize_beverlend_lemA.10_3}, we know that
\begin{equation}
C'''=C_0\otimes\ketbra00+C_1\otimes\ketbra11,
\end{equation}
where $C_0$ is an $(n+m-1)$-qubit Clifford.
Furthermore, by \Cref{eq:lem:generalize_beverlend_lemA.10_4}, we have
\begin{equation}
C_0\pbra{\ket{\phi_\lambda'}\ket{0^{r+m-1}}}=\ket{\psi_\lambda'}\ket{0^{r+m-1}}
\end{equation}
holds for all $\lambda\in\Scal$.

Iterating the above process to ensure $C_0$ commutes with the single-qubit $\Zgate$ gate on the last qubit, we can repeatedly trim off the trailing ancillas until we obtain
\begin{equation}
\tilde C\pbra{\ket{\phi_\lambda'}\ket{0^r}}=\ket{\psi_\lambda'}\ket{0^r}
\quad\text{holds for all $\lambda\in\Scal$,}
\end{equation}
where $\tilde C$ is some $n$-qubit Clifford.
Recall from \Cref{eq:lem:generalize_beverlend_lemA.10_1.1} and \Cref{eq:lem:generalize_beverlend_lemA.10_1.2} that $\ket{\phi_\lambda'}\ket{0^r}=C_\phi^\dag\ket{\phi_\lambda}$ and $\ket{\psi_\lambda'}\ket{0^r}=C_\psi^\dag\ket{\psi_\lambda}$.
Hence setting $C=C_\psi\tilde CC_\phi^\dag$ guarantees $C\ket{\phi_\lambda}=\ket{\psi_\lambda}$.
This completes the proof of \Cref{lem:generalize_beverlend_lemA.10}.
\end{proof}

\section{Improved analysis of Low-Kliuchnikov-Schaeffer state preparation}\label{sec:improve_LKS}

We observe that the proof of \Cref{thm:diagonal_unitary_T-count} works equally well for diagonals of single-qubit unitaries.
With \Cref{fct:unitary_syn_n=1}, we can also decompose arbitrary single-qubit unitary into diagonal unitaries for which \Cref{thm:diagonal_unitary_T-count} directly applies.
These are presented as the following \Cref{lem:block_diag_single_T-count}.

\begin{lemma}\label{lem:block_diag_single_T-count}
Let $n\ge0$ be an integer and $U_1,U_2,\ldots,U_{2^n}$ be arbitrary single-qubit unitaries.
Define $(n+1)$-qubit unitary $U=\diag(U_1,\ldots,U_{2^n})$.
Then $U$ can be implemented up to error $\eps$ by a Clifford+$\Tgate$ circuit using
$O\pbra{\sqrt{2^n\log(1/\eps)}+\log(1/\eps)}$ $\Tgate$ gates and ancillary qubits.
\end{lemma}
\begin{proof}
The proof is identical to the proof \Cref{thm:diagonal_unitary_T-count} once we make sure that the determinant of each $U_j$ is $1$.
To achieve this, let $D=\diag(\alpha_1,\beta_1,\ldots,\alpha_{2^n},\beta_{2^n})$ be an $(n+1)$-qubit diagonal unitary such that the determinant of $U_j'$ is $1$ for all $j$, where $U_j'=\diag(\alpha_j,\beta_j)U_j$.
Note that $D^\dag$ can be constructed using \Cref{thm:diagonal_unitary_T-count}, and then $DU=\diag(U_1',\ldots,U_{2^n}')$ can be handled in the same way as in \Cref{thm:diagonal_unitary_T-count}.
Combining both gives $D^\dag(DU)=U$.

An alternative proof is to use \Cref{fct:unitary_syn_n=1}, which represents each $U_j$ as $A_j\Hgate B_j\Hgate C_j$ for single-qubit diagonal unitaries $A_j,B_j,C_j$.
Then we have
\begin{align}
U
&=\diag(A_1,\ldots,A_{2^n})\diag(\Hgate,\ldots,\Hgate)\\
&\cdot\diag(B_1,\ldots,B_{2^n})\\
&\cdot\diag(\Hgate,\ldots,\Hgate)\diag(C_1,\ldots,C_{2^n}),
\end{align}
where $\diag(\Hgate,\ldots,\Hgate)$ is simply a Hadamard gate on the $(n+1)$-th qubit and, moreover, $\diag(A_1,\ldots,A_{2^n})$ (resp., $\diag(B_1,\ldots,B_{2^n}),\diag(C_1,\ldots,C_{2^n})$) can be implemented by \Cref{thm:diagonal_unitary_T-count}.
\end{proof}

As a direct application of \Cref{lem:block_diag_single_T-count}, we improve the state synthesis bound in \cite{low2018trading}.
Note that \Cref{cor:improve_lks} is only off by the extra $n$ in front of $\log(1/\eps)$, which does not dominate the overall complexity unless $\eps$ is extremely small.
Given this, this construction may be of practical interest.

\begin{corollary}\label{cor:improve_lks}
Any $n$-qubit state can be prepared up to error $\eps$ by a Clifford+$\Tgate$ circuit starting with the all-zeros state using
\begin{equation}
O\pbra{\sqrt{2^n\log(1/\eps)}+n\log(1/\eps)}
\text{ $\Tgate$ gates}
\end{equation}
and
\begin{equation}
O\pbra{\sqrt{2^n\log(1/\eps)}+\log(1/\eps)}
\text{ ancillary qubits.}
\end{equation}
\end{corollary}
\begin{proof}
The state synthesis algorithm in \cite{low2018trading} uses Grover-Rudolph algorithm \cite{grover2002creating}:
The first qubit is rotated to have the correct marginal. Then conditioned on the first qubit, the second qubit is rotated to have the correct marginal. This process is iterated  until the $n$th qubit is rotated conditioned on the first $n-1$ qubits.

Let $k=0,1,\ldots,n-1$.
The $k$-th operation above is a rotation on the $(k+1)$-th qubit, controlled by the first $k$ qubits.
Therefore, it is a diagonal of single-qubit unitaries, for which we apply \Cref{lem:block_diag_single_T-count}.
Let $\eps_k=\eps/2^{n-k}$.
Then by \Cref{lem:block_diag_single_T-count}, the $k$-th operation above can be $\eps_k$-approximately implemented with 
\begin{align}
O\pbra{\sqrt{2^k\log(1/\eps_k)}+\log(1/\eps_k)}
&=O\biggl(\sqrt{2^k\pbra{n-k+\log(1/\eps)}}+n-k+\log(1/\eps)\biggr)
\end{align}
T gates and ancillas.

Through all $k=0,1,\ldots,n-1$, the maximal number of ancillas is
\begin{align}
O\biggl(\max_k\sqrt{2^k\pbra{n-k+\log(1/\eps)}}
+n-k+\log(1/\eps)\biggr)
&=O\pbra{\sqrt{2^k\log(1/\eps)}+\log(1/\eps)}.
\end{align}
The total error is
\begin{equation}
\sum_{k=0}^{n-1}\eps_k=\eps\cdot\sum_{k=0}^{n-1}2^{k-n}\le\eps,
\end{equation}
and the total number of $\Tgate$ gates used is
\begin{align}
&\phantom{\le}\sum_{k=0}^{n-1}O\biggl(\sqrt{2^k\pbra{n-k+\log(1/\eps)}}
+n-k+\log(1/\eps)\biggr)\\
&\le\sum_{k=0}^{n-1}O\biggl(\sqrt{2^k(n-k)}+\sqrt{2^k\log(1/\eps)}+\log(1/\eps)\biggr)
\tag{since $\sqrt{a+b}\le\sqrt a+\sqrt b$}\\
&=O\pbra{\sqrt{2^n\log(1/\eps)}+n\log(1/\eps)}.
\tag*{\qedhere}
\end{align}
\end{proof}

\end{document}